\documentclass[11pt]{article}

\usepackage[margin=1in]{geometry}
\usepackage{amsmath,amssymb,amsthm}
\usepackage{placeins}
\usepackage{bm}
\usepackage{booktabs}
\usepackage{multirow}
\usepackage{graphicx}
\usepackage{subcaption}
\usepackage[ruled,vlined]{algorithm2e}
\usepackage[numbers]{natbib}
\usepackage{xcolor}
\usepackage{url}
\usepackage[colorlinks=true,citecolor=blue,linkcolor=blue,urlcolor=blue]{hyperref}
\usepackage{placeins}

\newtheorem{assumption}{Assumption}

\newtheorem{proposition}{Proposition}
\newtheorem{theorem}{Theorem}

\newcommand{\Description}[1]{}
\newenvironment{acks}{%
  \phantomsection
  \addcontentsline{toc}{section}{Acknowledgments}
  \section*{Acknowledgments}%
}{}

\let\origincludegraphics\includegraphics
\renewcommand{\includegraphics}[2][]{%
  \IfFileExists{#2}{\origincludegraphics[#1]{#2}}{%
    \fbox{\begin{minipage}[c][1.6in][c]{0.9\linewidth}\centering\footnotesize Missing figure file:\\ \texttt{\detokenize{#2}}\end{minipage}}%
  }%
}

\title{Causal Estimation of Share-Induced Engagement with Flywheel Effects\thanks{Accepted by the 32nd ACM SIGKDD Conference on Knowledge Discovery and Data Mining (KDD 2026).}}
\author{%
Weitao Cheng\\
\small HKUST\\
\small \texttt{wchengal@connect.ust.hk}
\and
Yilin Li\\
\small Tencent\\
\small \texttt{ilynli@tencent.com}
\and
Yong Wang\\
\small Tencent\\
\small \texttt{darwinwang@tencent.com}
\and
Nian Si\thanks{Corresponding author.}\\
\small HKUST\\
\small \texttt{niansi@ust.hk}\\
}
\date{}

\begin{document}
\maketitle

\begin{abstract}
Sustainable user growth in online platforms depends not only on acquiring new users but also on reactivating and engaging existing ones through social sharing features. A well-designed sharing feature can trigger a self-reinforcing ``flywheel effect'': reactivated users become potential sharers whose engagement propagates through the network over multiple rounds, amplifying total engagement. Measuring the causal impact of such sharing features is challenging, as their effects unfold through complex social networks and temporal cascades, violating the no-interference assumption underlying classical A/B testing. We develop a framework for experiments on sharing features that accounts for interference caused by the flywheel effect and targets a global treatment effect on share-induced engagement. Our estimator is motivated by a flow-balance identity and interprets share-induced engagement as a geometric amplification process, yielding a closed-form propagation adjustment that accounts for multi-round diffusion using commonly available attribution logs.  Under mild conditions, we establish consistency of the proposed estimator and develop a valid A/A testing procedure for pipeline validation. Simulation studies show that our method substantially reduces bias relative to the difference-in-means estimator and first-order adjustments, while the proposed A/A test maintains nominal Type I error. We also extend the framework to a user-level reactivation metric via a Poisson approximation.  Finally, we demonstrate the approach on a real-world large-scale online platform and discuss empirical implications for evaluating sharing feature designs.
\end{abstract}


\bigskip
\section{Introduction}
\label{sec:intro}
The pursuit of sustainable user growth is a central objective for online content platforms. Beyond acquiring new users through advertising or promotions, an increasingly important strategy is to reactivate dormant users and deepen engagement among existing ones through the careful design of recommendation systems. Among these mechanisms, social sharing features play a particularly critical role by enabling users to disseminate content to peers within and beyond the platform \cite{heimbach2018sharingmechanism}. When designed effectively, such features can generate a self-reinforcing feedback loop—what we term the \textbf{flywheel effect} (borrowing the flywheel metaphor from the management literature; see \citep{anderson2022value}). In this loop, shared content re-engages inactive recipients, who subsequently become active users and potential future sharers \citep{zhao2015seismic, rizoiu2017tutorial}.
Through this mechanism, engagement and user growth propagate organically through the underlying social network, thereby amplifying the long-run impact of an initial product intervention.

Engineers on recommendation-system platforms frequently propose new features aimed at boosting content sharing. These interventions may involve modifying recommendation algorithms to promote content that is more likely to be shared, or redesigning the sharing interface itself to reduce friction and make sharing easier for users. Like other product features, all such changes must be evaluated through A/B testing before being deployed.

However, precisely quantifying the causal effect of introducing a new sharing feature is challenging using classical A/B testing methodologies. Unlike direct treatments with immediate, individual-level responses, the effects of sharing features are inherently indirect, network-mediated, and dynamic, often unfolding over extended time horizons through cascades of social interactions and feedback loops. In such settings, the presence of interference violates the Stable Unit Treatment Value Assumption (SUTVA) \citep{imbens2015causal}. As a consequence, standard experimental designs may substantially mismeasure—or even fail to detect—the true impact of these features on user growth and engagement.

In this paper, we study experiments involving sharing features and develop a theoretical framework to analyze the interference induced by flywheel effects. Building on this framework, we propose several practical and easy-to-implement estimators for A/B tests that achieve low bias and low variance, enabling reliable quantification of the causal impact of sharing features across a range of metrics of interest. Our main contributions are as follows:

\textbf{1. A framework for flywheel effects.} We provide a model of sharing dynamics on online platforms that captures how content sharing induces network-mediated interference and dynamic feedback loops via a multivariate Hawkes process \citep{embrechts2011multivariate}. The model explicitly characterizes how an initial intervention can propagate through user reactivation and subsequent sharing, thereby enabling a principled analysis of different estimators in A/B tests under flywheel effects.

   \textbf{2. A propagation-adjusted treatment effect estimator.} Based on our model, we derive a sender--receiver flow-balance identity that links share-induced views with receiver-side share-driven views. Building on this identity, we propose a class of estimators for various metrics of interest. The estimators are log-based and model-free, and do not require fitting Hawkes process parameters, making them robust to model misspecification.

 \textbf{3. Theoretical guarantees for the proposed estimators and inference procedures.} Under a homogeneous propagation setting and
 % mild local dependence conditions
 mild conditions, we establish consistency of the proposed estimator. We further develop an inference procedure that achieves asymptotically exact Type I error control in A/A tests.

 \textbf{4. Numerical simulation and real-world deployment performance.} 
We evaluate our estimators through numerical simulations. The results show that our methods achieve the lowest bias and comparable variance among all benchmark approaches.  
Furthermore, the proposed method has been deployed on the experimentation platform of a large-scale real-world social network. Several sharing-related metrics have been developed and integrated with multiple business applications, successfully supporting the evaluation of updated sharing features and algorithms.

The remainder of the paper is organized as follows. Section \ref{sec:literature} reviews the related literature. Section \ref{sec:sharingmodel} introduces the content-sharing framework. Section \ref{sec:AB} studies the experimental design and presents our proposed estimator. Section \ref{sec:inference} develops the corresponding inference procedure. Section \ref{sec:GTE_ra} discusses the extension to reactivation rate. Finally, Sections \ref{sec:numerical} and \ref{sec:realworld} report the simulation results and real-world implementation, respectively.

% \vspace{-0.15in}
\section{Literature Review}
\label{sec:literature}
In this section, we review related work on sharing mechanisms and interference in online platforms and network settings.

\textbf{Sharing Mechanism.}
A large body of literature studies the \emph{sharing mechanism}—the product design choices that determine friction and available sharing channels. \citet{heimbach2018sharingmechanism} show that sharing-mechanism design can substantially alter user sharing behavior, highlighting the need for rigorous evaluation of sharing-feature changes.

At the process level, resharing naturally induces a self-reinforcing cascade: each new reshare exposes the content to additional audiences who may further reshare it, leading to multi-round amplification and a flywheel effect. This compounding dynamic is well captured by self-exciting point-process models. \citet{zhao2015seismic} model reshare arrivals as a self-exciting process and leverage early diffusion signals to predict eventual popularity, while \citet{rizoiu2017tutorial} provide a comprehensive overview of Hawkes-process foundations and practical estimation tools for modeling diffusion and contagion in social media.

Field experiments complement these modeling approaches by quantifying peer-mediated effects of social exposure. \citet{bakshy2012social} present experimental evidence from social advertising showing that peer exposure can propagate treatment effects beyond directly treated users. Relatedly, the influence maximization literature studies how seeding strategies shape cascade reach; \citet{su2023unveiling} analyze the sensitivity of individual gains in such settings, offering an optimization perspective on amplification in sharing flywheels.

\noindent\textbf{Experimental Design and Bias Analysis in Online Platforms.} Interference is widely recognized in the literature, and empirical evidence
\citep{blake2014marketplace,holtz2020reducing,fradkin2015search} shows that the bias it induces can be comparable in magnitude to the treatment effect itself. In particular, A/B tests on two-sided platforms are especially vulnerable to interference arising from competition and spillover effects. \citet{johari2022experimental}, \citet{li2022interference}, and \citet{dhaouadi2023price} study bias on both the demand and supply sides using stochastic modeling frameworks. Building on this line of work, \citet{johari2022experimental, bajari2023experimental}, and \citet{masoero2026multiple} propose two-sided randomization schemes, also referred to as multiple-randomization designs. In addition, \citet{bright2025reducing} analyze linear-programming-based matching mechanisms and develop debiased estimators using shadow prices, while \citet{wang2026experimental} further investigates interventions that directly modify the matching mechanism itself.
Other work focuses on more specialized settings. For example, \citet{chawla2016b,basse2016randomization,liao2023statistical,liao2024interference} study experimental design in auctions; \citet{si2023tackling} examines interference induced by feedback loops; \citet{zhan2024estimating,zhang2025debiasing} consider seller-side experiments under recommender-system interference; \citet{li2025causal,li2025dyadic} focus on dyadic data settings; and \citet{wager2021experimenting,munro2025treatment} study experimentation under equilibrium models. However, to the best of our knowledge, interference induced by sharing mechanisms has not been systematically studied in the existing literature.

\noindent\textbf{Network Interference. }Interference on networks has also been extensively investigated across operations research, statistics, and economics \citep{hudgens2008toward, gui2015network, basse2019randomization, li2022random, comola2021treatment,shirani2024causal}. A common strategy for reducing spillover effects is cluster randomization, where units are partitioned into approximately disconnected subgraphs using prior knowledge of the network topology; see, for example, \citet{aronow2017estimating, candogan2024correlated, ugander2013graph, ugander2023randomized}.
In platform environments with endogenous interactions, \citet{weng2024experimental} analyze settings in which treatment assignments influence outcomes through algorithmic matching mechanisms, such as opponent selection in online games. Another stream of work adopts parametric interference structures — often linear — which permits identification and estimation using regression-based methods \citep{toulis2013estimation, baird2018optimal, bramoulle2009identification, harshaw2023design, yu2022estimating}. 
A closely related line of work summarizes network interference through exposure mappings.  
\citet{forastiere2022estimating} estimate direct and spillover effects using Bayesian generalized propensity scores.
\citet{leung2022causal} studies randomized experiments under approximate neighborhood interference, allowing potential outcomes to depend on the full treatment assignment vector and allowing treatments assigned to distant alters to have nonzero but decaying effects.
These frameworks are close in spirit to
the diffusion-like spillovers considered in this paper. The key distinction is the target estimand and the resulting estimator: these works primarily focus on exposure-specific effects defined through exposure mappings, whereas we target the platform-level global treatment effect induced by full deployment of a sharing feature and develop a propagation-adjusted estimator tailored to attribution logs and multi-round sharing cascades. 
Another related but distinct literature studies group-formation experiments, where the intervention directly determines how individuals are matched and the central objective is to measure peer effects \citep{basse2024randomization}. 

\section{A Framework of Content Sharing}
\label{sec:sharingmodel}

We develop a dynamic model for sharing events in social networks using a multivariate Hawkes process \cite{embrechts2011multivariate, rizoiu2017tutorial}. Based on this model, we define the quantities of interest and formalize the notion of the GTE.

\subsection{Multivariate Hawkes Process for Sharing}
\label{subsec:hawkes}

We consider an online platform with
users indexed by $i\in\{1,\dots,n\}$ and contents indexed by $k\in\{1,\dots,m\}$.
User relations are represented by an undirected graph $\mathcal{G}=(\mathcal{V},\mathcal{E})$
with adjacency matrix $G=(g_{ij})_{n\times n}$, where $g_{ij}=1$ if users $i\neq j$ are connected
and $g_{ii}=0$. Let $\bar g := n^{-1}\sum_{i,j} g_{ij}$ denote the average degree of $\mathcal{G}$.

For user--content interactions, for each content $k$ we define an $n$-variate marked point process
$N^{(k)}(t)=\bigl(N_{1k}(t),\dots,N_{nk}(t)\bigr), t \geq 0$, 
where $N_{ik}(t)$ counts the number of view events of content $k$ by user $i$ up to time $t$.
Let $\mathcal{H}_t$ denote the filtration generated by all events up to time $t$. 
We decompose the view process into two components according to its origin: self-discovery and sharing:
\[
N_{ik}(t)=N_{ik}^{d}(t)+N_{ik}^{s}(t),
\]
where $N^{d}_{ik}(t)$ counts views arising from independent discovery channels
(e.g., recommendation exposure, search) and $N^{s}_{ik}(t)$ counts views triggered
by sharing. 
Suppose that events in component $(i,k)$ occur at times $\{t^{d}_{ik\ell}\}_{\ell\ge 1}$ for self-discovery
and $\{t^{s}_{ik\ell}\}_{\ell\ge 1}$ for sharing. Then we can write
\begin{equation}
N_{ik}^{d}(t) = \sum_{\ell\ge 1}\mathbf{1}\{t^{d}_{ik\ell}\le t\}, \text{ and } N_{ik}^{s}(t) = \sum_{\ell\ge 1}\mathbf{1}\{t^{s}_{ik\ell}\le t\}.
\label{eq:N_sum_ds}
\end{equation}

Content viewing and sharing on social platforms exhibit temporally clustered, self-reinforcing dynamics: a view or share exposes neighbors and transiently increases the likelihood of subsequent views or reshares, with effects that decay over time. Self-exciting point processes—particularly Hawkes processes—are well suited to capture such short-term excitation and decay, and have been widely used to model information diffusion and popularity dynamics in social media (e.g., retweet cascades and item popularity) \citep{zhao2015seismic, rizoiu2017tutorial, rizoiu2017hip}. Moreover, the branching-process interpretation provides an intuitive view of content diffusion by separating exogenous discovery from endogenous social propagation.

Specifically, we model the discovery component $N_{ik}^{d}(t)$ as a Poisson process with a nonhomogeneous, nonnegative intensity $\mu_{ik}(t)$, which may vary across users and contents~\cite{farajtabar2014shaping}. We assume that $\mu_{ik}(t)$ has support only on $[0,T]$, reflecting the typically short duration of the experiment. 
For each ordered pair $i\to j$ and content $k$, we define two nonnegative \emph{marked kernels}, $\phi_{ijk}^{d}(t)$ and $\phi_{ijk}^{s}(t)$ for $t\ge 0$, corresponding to self-discovery and sharing origins, respectively. We impose the network constraint that 
$
\phi_{ijk}^{d}(t)=\phi_{ijk}^{s}(t)\equiv 0$, whenever $(i,j)\notin \mathcal{E}.
$
Here, if a user shares a piece of content, we assume that the sharing occurs exactly at the associated view time $t$.

The share-induced viewing intensity for receiver $j$ viewing content $k$, namely the conditional intensity function of the process $N^{s}_{jk}(t)$, is given by
\[
\lambda^s_{jk}(t\mid \mathcal{H}_{t^-})
=
\sum_{i=1}^n \left( \sum_{\ell: t^s_{ik\ell}<t} \phi^s_{ijk}(t-t^s_{ik\ell}) + \sum_{\ell: t^d_{ik\ell}<t} \phi^d_{ijk}(t-t^d_{ik\ell}) \right) ,
\]
and thus, the total view intensity,  namely the conditional intensity function of the process $N_{jk}(t)$, is given by
\[
\lambda_{jk}(t\mid \mathcal{H}_{t^-})
=
\mu_{jk}(t) + \lambda^s_{jk}(t\mid \mathcal{H}_{t^-}).
\]
A commonly used parameterization factorizes the kernel into an edge-strength term and a temporal decay term:
\[
\phi^d_{ijk}(t)
=
d_{ijk}^d\cdot \omega_{ijk}(t),\text{ and }\phi^s_{ijk}(t)
=
d_{ijk}^s\cdot \omega_{ijk}(t),
\quad t > 0,
\]
where
 $d^d_{ijk}\ge 0,d^s_{ijk}\ge 0$ are edge-content specific propagation strengths (and $d^d_{ijk}=d^s_{ijk}=0$ if $(i,j)\notin \mathcal{E}$) and $\omega_{ijk}:\mathbb{R}_+\to\mathbb{R}_+$ is a decay kernel governing response-time lags, which usually concentrates on the low end. Without loss of generality, we assume the kernel is normalized such that
$\int_0^\infty \omega_{ijk}(t)dt = 1$ and hence,
\[
\int_0^\infty  \phi^d_{ijk}(t)\,dt=d^d_{ijk} \text{ and } \int_0^\infty  \phi^s_{ijk}(t)\,dt=d^s_{ijk}.
\]
Therefore, $d_{ijk}^{d}$ and $d_{ijk}^{s}$ can be interpreted as the expected number of share-induced views at user $j$ generated by a single view event at user $i$ for content $k$, arising from self-discovery and sharing origins, respectively.

For each content $k$, define the dominating branching matrix
\[
Q^{(k)} := \Bigl(\max\left \{d^d_{ijk},d^s_{ijk}\right \}\Bigr)_{i,j=1}^n,
\]
which serves as an entrywise upper bound on the mean share-induced view counts across users.
A standard sufficient condition
for stability is that the spectral radius $\rho(Q^{(k)})<1$ for each $k$ \cite{daley2003introduction, liniger2009multivariate, achab2018uncovering}.

It is well known that a Hawkes process admits a cluster representation \cite{hawkes1974cluster}: it can be viewed as a Poisson cluster process generated by an immigration–offspring mechanism. 
Specifically, immigrant events occur on the timeline according to a nonhomogeneous Poisson process $N_{ik}^d(\cdot)$. Each immigrant initiates a cluster and independently generates child events after its occurrence time according to another nonhomogeneous Poisson process with rate determined by $\phi^d_{ijk}(\cdot)$. Each child event then acts as a new parent and produces its own offspring according to the same rule with $\phi^s_{ijk}(\cdot)$, independently of all other events, so descendants are formed recursively across generations. All events descended from the same immigrant constitute a single cluster. Under this construction, the resulting process is equivalent in law to the intensity-based formulation $N_{ik}(\cdot)$.

We remark that in large-scale platforms, both the number of users $n$ and the number of contents $m$ can be large.
Nevertheless, each user typically has  $O(1)$ share-induced and share-driven view events within the observation window. The corresponding conditions are stated precisely in Assumption~\ref{assum:bounded}.

\subsection{The Impact of Sharing and Global Treatment Effect}
We are primarily interested in measuring the impact of sharing, defined as
\begin{equation}
    \mathrm{IS}
    :=
    \frac{1}{n}
    \mathbb{E}\!\left[
        \sum_{i=1}^n \sum_{k=1}^m
        \int_0^\infty O_{ik}\, dN_{ik}^s(t)
    \right],
    \label{def:IS}
\end{equation}
where $O_{ik}$ denotes a metric of interest associated with user $i$ and content $k$, independent of all other random variables. 
For example, if $O_{ik}=1$, then $\mathrm{IS}$ represents the expected number of sharing-induced views; 
if $O_{ik}$ is the viewing time of user $i$ on content $k$, then $\mathrm{IS}$ represents the expected sharing-induced watching time.  For convenience of exposition, we focus on the case $O_{ik}=1$ in the following development; however, all results and estimators extend directly to a general outcome variable $O_{ik}$.

Let $D^d_{ijk}(t),D^s_{ijk}(t)$ denote the number of share-induced view events resulting from user $i$ sharing content $k$ with user $j$  at time $t$,  arising from self-discovery and sharing origins, respectively.  
 Hence, conditioning on the occurrence of the corresponding share-induced view event at time $t$, the model of Section \ref{subsec:hawkes} gives us 
 $$
 \mathbb{E}[D^v_{ijk}(t)\mid \Delta N_{ik}^{v}(t) = 1]=d^v_{ijk}, \quad v \in \{d,s\},
 $$
 where $\Delta N_{ik}^{d}(t)=1$ (or $\Delta N_{ik}^{s}(t)=1$) means that the counting process $N_{ik}^{d}$ (or $N_{ik}^{s}$) has a unit jump at time $t$. By the cluster (branching) representation of the Hawkes process, each share–induced view event corresponds to exactly one sender and one viewer. Hence, we have the following flow-balance identity:
\begin{align}
        &\sum_{i=1}^n   \sum_{j=1}^n \sum_{k=1}^m\left(
        \int_0^\infty  D^d_{ijk}(t)\, dN^d_{ik}(t)+
        \int_0^\infty  D^s_{ijk}(t)\, dN^s_{ik}(t)\right) \notag \\
        =&\sum_{i=1}^n \sum_{k=1}^m
        \int_0^\infty 1\cdot\, dN_{ik}^s(t),
        \label{eq:flow-balance}
\end{align}
where the left-hand side counts offspring view events from the sender perspective, while the right-hand side counts the same events from the receiver perspective. Since each share-induced view event involves exactly one sender and one viewer, both sides enumerate the same set of parent–child edges, and thus the identity holds. The following assumption is designed to align with real-world settings.

\begin{assumption} We assume
    \begin{enumerate}
        \item There exists \(T>0\) such that the exogenous intensity satisfies
    \(\mu_{ik}(t)=0\) for all \(t>T\) and all \(i,k\). Moreover, there exists
    a constant \(B<\infty\) such that, for all \(n\),
    \[
    \max_{1\le i\le n}\sum_{k=1}^m \int_0^\infty \mu_{ik}(t)\,dt \le B,
    \ \ 
    \max_{1\le j\le n}
    \sum_{k=1}^m\sum_{i=1}^n w_{ik}^d d_{ijk}^d \le B,
    \]
    where
    \(w_{ik}^d=\mathbb E[N_{ik}^d(\infty)]
    =\int_0^\infty \mu_{ik}(t)\,dt\).

    \item There exists a constant \(\bar{c}\in(0,1)\) such that, for all \(n\),
    \[
    \max_{1\le k\le m}\max_{1\le i\le n}
    \sum_{j=1}^n \max\{d_{ijk}^d,d_{ijk}^s\}\le \bar{c},
    \ \ 
    \max_{1\le j\le n}\sum_{i=1}^n
    \max_{1\le k\le m} d_{ijk}^s \le \bar{c} .
    \]

    \item There exists a constant $\underline{c} > 0$ such that, for all $n$
    \[
    \frac{1}{n} \sum_{i=1}^n   \sum_{j=1}^n \sum_{k=1}^m
        w_{ik}^{d} d_{ijk}^{d}  \ge \underline{c}.
    \] 
    \end{enumerate}
    \label{assum:bounded}
\end{assumption}

    Conditions (1) and (2) ensure that the expected numbers of discovery-origin and share-induced events per user remain uniformly bounded.  Condition (1) bounds the expected number of discovery-origin views per user and the amount of direct sharing traffic received by any single user. Condition (2) keeps each content-specific sharing cascade subcritical and prevents sharing traffic from concentrating on any single user. Condition (3) prevents the setting from degenerating into a trivial no-sharing case.
    Then, the impact-of-sharing metric is well defined.

\begin{proposition}[Well-definedness of $\mathrm{IS}$]
    Under Assumption \ref{assum:bounded} (1), (2) and $O_{ik} =1$, the impact of sharing $\mathrm{IS}$ is well defined, i.e., $\mathrm{IS}<\infty$.
    \label{prop:IS-finite}
\end{proposition}
The proof of Proposition \ref{prop:IS-finite} is deferred to Appendix \ref{subsec:proof-of-IS-finite}.

We assume that the intervention operates by changing the sender’s behavior through the sharing strength parameters $d^d_{ijk}$ and $d^s_{ijk}$. 
Let
$d_{ijk}^{v,T},v\in\{d,s\}$ and $d_{ijk}^{v,C},v\in\{d,s\}$ denote the sharing strength under treatment and
control, respectively. We define $\mathrm{IS}^{T}$ and $\mathrm{IS}^{C}$ as the
resulting sharing impact under treatment and control, induced by
$d_{ijk}^{v,T},v\in\{d,s\}$ and $d_{ijk}^{v,C},v\in\{d,s\}$ for all users on the platform, respectively. We further assume that Assumption \ref{assum:bounded} holds under both global treatment and global control regimes so that \(\mathrm{IS}^{T}\) and \(\mathrm{IS}^{C}\) are well-defined.
The global treatment effect (GTE)
is then defined as
\[
\mathrm{GTE}=\mathrm{IS}^{T}-\mathrm{IS}^{C}.
\]
In the following, we describe how to estimate this quantity using A/B tests. 
A key difficulty in estimating the GTE is that a change in $d_{ijk}$ can have
long-term effects on other users through network propagation. In particular,
the observed sharing behavior of user $i$ may depend on the sharing strengths
of other users. Therefore, under an experimental assignment in A/B tests, sharing behaviors
generally depend on the entire treatment assignment vector, which violates the
stable unit treatment value assumption (SUTVA) \citep{imbens2015causal}.

\section{A/B Tests and Estimators}

\label{sec:AB}

\subsection{Bernoulli Experiment}
\label{subsec:experimentation}
We consider a Bernoulli experimental design in a layered experimentation system, which is arguably one of the most widely used designs in practice. 
We introduce an exposure indicator $V_i$, where $V_i=1$ means user $i$ is included in the experiment and $V_i=0$ otherwise, with $\mathbb{P}(V_i=1)=\pi$ for some $0<\pi\le 1$. Conditional on $V_i=1$, treatment is assigned with probability $\mathbb{P}(Z_i=1 \mid V_i=1)=p$ for some $0<p<1$; users with $V_i=0$ are not exposed and we set $Z_i=0$. The pairs $\{(V_i,Z_i)\}_{i=1}^n$ are independent across users.
Let $\mathbf{V}_n=(V_1,\dots,V_n)^\top$ denote the exposure vector and $\mathbf{Z}_n=(Z_1,\dots,Z_n)^\top$ denote the assignment vector. Define the marginal probabilities
$
\pi_T:=\mathbb P(Z_iV_i=1)=\pi p, \pi_C:=\mathbb P((1-Z_i)V_i=1)=\pi(1-p).
$
Define the realized treatment and control group sizes within the experiment as
$
n_T:=\sum_{i=1}^n Z_iV_i, \ 
n_C:=\sum_{i=1}^n (1-Z_i)V_i.
$
When $\pi=1$, all users are included in the experiment. 

Recall the intervention changes the sender's behavior. Therefore, if user $i$ is assigned to treatment, then $d_{ijk}^{d}=d_{ijk}^{d,T}$ and $d_{ijk}^{s}=d_{ijk}^{s,T}$; if user $i$ is assigned to control, then $d_{ijk}^{d}=d_{ijk}^{d, C}$ and $d_{ijk}^{s}=d_{ijk}^{s,C}$.

\subsection{Estimation}
We derive our estimator from the flow-balance identity \eqref{eq:flow-balance}. By observing 
$N_{ik}(t) = N^{d}_{ik}(t) + N^{s}_{ik}(t)$, the flow-balance identity \eqref{eq:flow-balance} can be equivalently written as 
\begin{align}
        &\sum_{i=1}^n   \sum_{j=1}^n \sum_{k=1}^m
        \int_0^\infty D^d_{ijk}(t)\, dN^d_{ik}(t) \notag \\
        =&\sum_{i=1}^n \sum_{k=1}^m
        \int_0^\infty \left(1- \sum_{j=1}^n D^s_{ijk}(t)\right)\, dN_{ik}^s(t).
        \label{eq:flow-balance2}
\end{align}

Taking expectations on both sides of \eqref{eq:flow-balance2}, we obtain the identity in the following proposition.

\begin{proposition}[Expectation form of the flow-balance identity]
\label{prop:refined-flow-balance-param-form}
Taking expectations on both sides of \eqref{eq:flow-balance2} yields
\begin{equation}
\sum_{i=1}^{n}  \sum_{j=1}^{n} \sum_{k=1}^{m}   w_{ik}^{d} d^d_{ijk}
= \sum_{i=1}^{n} \sum_{k=1}^{m}   w_{ik}^{s} \Bigl( 1 - q_{ik}\Bigr),
\label{eq:refined-flow-balance-param-form}
\end{equation}
where $q_{ik}=\sum_{j=1}^n d^s_{ijk}$ and
$$w_{ik}^d=\mathbb{E}[N^d_{ik}(\infty)]=\int_0^\infty \mu_{ik}(t)\,dt,\quad w_{ik}^s=\mathbb{E}[N^s_{ik}(\infty)].$$
\end{proposition}

The proof of Proposition \ref{prop:refined-flow-balance-param-form} is deferred to Appendix \ref{subsec:proof-of-flow-balance-param-form}.

Equation~\eqref{eq:refined-flow-balance-param-form} holds under global treatment
(i.e., $d^v_{ijk}=d_{ijk}^{v,T}$ for $v\in\{s,d\},i=1,2,\ldots,n$ with $q_{ik}=q^T_{ik}$ defined accordingly), global control
(i.e., $d^v_{ijk}=d_{ijk}^{v,C}$ for $v\in\{s,d\},i=1,2,\ldots,n$ with $q_{ik}=q^C_{ik}$ defined accordingly), and under the experimental
setting conditional on the treatment assignment.

Note that the impact of sharing is written as 
$$
\mathrm{IS} =\frac{1}{n}\sum_{i=1}^{n} \sum_{k=1}^{m}   w_{ik}^{s}. 
$$
Combining this expression with Equation \eqref{eq:refined-flow-balance-param-form}, we have there exists an \textit{effective} downstream sharing rate $q_{\mathrm{eff}}$ such that 
$$
\mathrm{IS}= \frac{\frac{1}{n}\sum_{i=1}^{n}  \sum_{j=1}^{n} \sum_{k=1}^{m}   w_{ik}^{d} d^d_{ijk}}{1-q_{\mathrm{eff}}},
$$
and further
$$
\mathrm{GTE}=\frac{\frac{1}{n}\sum_{i=1}^{n}  \sum_{j=1}^{n} \sum_{k=1}^{m}   w_{ik}^{d} d^{d,T}_{ijk}}{1-q_{\mathrm{eff}}^T}-\frac{\frac{1}{n}\sum_{i=1}^{n}  \sum_{j=1}^{n} \sum_{k=1}^{m}   w_{ik}^{d} d^{d,C}_{ijk}}{1-q_{\mathrm{eff}}^C}.
$$

This representation admits a natural \emph{branching} interpretation: the numerator is the expected number of share-view offspring directly attributable to discovery-origin parent events, while $q_{\mathrm{eff}}$ is an average downstream sharing rate. Sharing cascades lead to a geometric amplification with factor $q_{\mathrm{eff}}$, yielding the multiplier $1/(1-q_{\mathrm{eff}})$.

This representation motivates the following estimation method. Let $Y_i^d$ denote the total number of share-induced view events generated by user $i$'s independent discovery-origin view events. Then
\[
\mathbb{E}[Y_i^d] = \sum_{j=1}^{n} \sum_{k=1}^{m} w_{ik}^{d}\, d^d_{ijk}.
\]
We define the sample means of $Y^d_i$ in the treatment and control groups, respectively, as
\begin{equation}
\widehat Y^d_T := \frac{1}{n_T}\sum_{i=1}^n Z_i V_i Y_i^{d},
\quad
\widehat Y^d_C := \frac{1}{n_C}\sum_{i=1}^n (1-Z_i) V_i Y_i^{d},
\label{eq:arm-specific-mean}
\end{equation}
to be the estimators for the numerators. 
To estimate $q^{T}_\mathrm{eff}$ and $q^{C}_\mathrm{eff}$, we use natural plug-in estimators. Specifically, we define $W_i^{s}$, the total number
of share-induced views received by user $i$, and $Y_i^{s}$, the total number of direct share-induced view events generated by user $i$'s share-origin view events. We estimate the downstream sharing rate in the treatment and control groups  by
\begin{equation}
    \begin{aligned}
        \widehat q^{T}
        =
        {\widehat Y^s_T}/
             {\widehat W^s_T },
        \quad
        \widehat q^{C}
        =
        {\widehat Y^s_C}/
             { \widehat W^s_C },
    \end{aligned}
    \label{eq:calculate-q}
\end{equation}
where 
\begin{align*}
\widehat Y^s_T &:= \frac{1}{n_T}\sum_{i=1}^n Z_i V_i Y_i^{s},
\quad
\widehat Y^s_C := \frac{1}{n_C}\sum_{i=1}^n (1-Z_i)  V_i Y_i^{s}, \\
\widehat W^s_T &:= \frac{1}{n_T}\sum_{i=1}^n  Z_i V_i W_i^{s},
\quad
\widehat W^s_C := \frac{1}{n_C}\sum_{i=1}^n (1-Z_i)  V_i W_i^{s}.
\end{align*}
Combining these estimators with the averages of share-view events directly induced by independent discovery channels  in \eqref{eq:arm-specific-mean} yields the estimator
\begin{equation}
\widehat{\mathrm{GTE}}
=
\frac{\widehat Y^d_T}{1-\widehat q^{T}}
-
\frac{\widehat Y^d_C}{1-\widehat q^{C}}.
\label{eq:prop-estimator}
\end{equation}

\subsection{Consistency of the Proposed Estimator}

In this section, we establish the consistency of the proposed estimator for the
global treatment effect (GTE).
For clarity of exposition, we first present a sufficient homogeneity condition.

\begin{assumption}[Homogeneous downstream view]\label{assum:homogeneity}
For each regime $a\in\{T,C\}$, there exists a constant $q^{a}$ such that for all $(i,k)$,
\[
q_{ik}^{a}\equiv q^{a}.
\]
\end{assumption}

Then, Assumption \ref{assum:homogeneity} together with Assumption \ref{assum:bounded} gives the following consistency results.

\begin{theorem}[Consistency of $\widehat{\mathrm{GTE}}$]
Suppose that Assumption~\ref{assum:bounded} holds under both global treatment and global control, and that Assumption~\ref{assum:homogeneity} holds. Then 
\[
\widehat{\mathrm{GTE}} - \mathrm{GTE}
\xrightarrow{p} 0.
\]
\label{theorem:consistency}
\end{theorem}
Detailed proof can be found in Appendix \ref{subsec:proof-theorem-consistency}.

Although Assumption \ref{assum:homogeneity} is stronger than strictly necessary, we argue that in most practical settings our estimator $\widehat{\mathrm{GTE}}$ remains nearly unbiased, as evidenced by both our simulation results and real-world studies. We believe this occurs because the estimation of $\widehat{q}$ involves aggregation across units and periods, which dampens the effect of localized violations of the homogeneity assumption.

\section{Inference Procedure}
\label{sec:inference}

In this section, we develop an inference procedure to construct confidence intervals for $\widehat{\mathrm{GTE}}$ and to conduct hypothesis testing. It is well known that inference for A/B tests under interference is challenging due to the dependence structure induced among units and between treatment assignments and outcomes, as well as estimator bias. Therefore, we establish a central limit theorem (CLT) for $\widehat{\mathrm{GTE}}$ under A/A tests, which enables asymptotically valid control of the type-I error.
\paragraph{A/A tests.}
Under A/A tests, we have $d_{ijk}^{s,T}=d_{ijk}^{s,C},d_{ijk}^{d,T}=d_{ijk}^{d,C}$, for all $i,j=1,2,\ldots,n,k=1,2,\ldots,m$.

For notational convenience, define the per-user vector
$
\bm U_i := (Y_i^{d},\, Y_i^{s},\, W_i^{s})^\top.
$ 
Let
$
n_V:=n_T+n_C,\ 
\bar{\bm U}_V:= n_{V}^{-1} \sum_{i=1}^n V_i\bm U_i,
$
and define the finite-population covariance matrix among exposed users by
\[
\Sigma_V
:=
\frac1{n_V-1}
\sum_{i=1}^n
V_i(\bm U_i-\bar{\bm U}_V)(\bm U_i-\bar{\bm U}_V)^\top.
\]
Let
$
\gamma_V:=\nabla g(\bar{\bm U}_V).
$
The corresponding finite-population standard error is
\begin{equation}
\operatorname{se}^2
:=
\left(\frac1{n_T}+\frac1{n_C}\right)
\gamma_V^\top\Sigma_V\gamma_V
=
\frac{n_V}{n_Tn_C}\gamma_V^\top\Sigma_V\gamma_V.
\label{eq:population-se}
\end{equation}

For $a\in\{T,C\}$, let
\[
\widehat{\bm U}_{a}
:= (\widehat{Y}^d_a,\,\widehat{Y}^s_a,\,\widehat{W}^s_a)^\top .
\]
We define $\widehat{\Sigma}_{T}$ as the sample covariance matrix of $\bm U_{i}$ in the treated samples,  
\begin{equation}
    \begin{aligned}
        \widehat{\Sigma}_{T}
        := \frac{1}{n_{T}-1}\sum_{i=1}^n Z_{i} V_i(\bm{U}_i-\widehat{\bm{U}}_{T})(\bm{U}_i-\widehat{\bm{U}}_{T})^\top,
    \end{aligned}
    \label{eq:sample-cov}
\end{equation}
and $\widehat{\Sigma}_{C}$ accordingly. By defining
$
g(u_{1},u_{2},u_{3})
:=
\frac{u_{1}}{1-u_{2}/u_{3}},
$
the GTE estimator can be equivalently written as 
$
\widehat{\mathrm{GTE}}
=
g(\widehat{\bm U}_{T})-g(\widehat{\bm U}_{C}).
$
We define the standard error under A/A tests as
\begin{equation}
    \begin{aligned}
        \widehat{\mathrm{se}}:=\sqrt{n_T^{-1}\nabla g(\widehat{\bm U}_T)^{\top}\widehat{\Sigma}_T \nabla g(\widehat{\bm U}_T)+n_C^{-1}\nabla g(\widehat{\bm U}_C)^{\top} \widehat{\Sigma}_C \nabla g(\widehat{\bm U}_C)},
    \end{aligned}
    \label{eq:se}
\end{equation}
where \(\nabla g(\widehat{\bm U}_T)\) denotes the gradient of \(g\) evaluated at
\(\widehat{\bm U}_T\), with \(\nabla g(\widehat{\bm U}_C)\) defined accordingly.

The following theorem justifies the normal approximation.

\begin{theorem}[Asymptotic normality under A/A tests]
\label{theorem:aa-normality}
Suppose that Assumption~\ref{assum:bounded} holds. Consider the A/A testing
scenario, i.e.,
\(d_{ijk}^{s,T}=d_{ijk}^{s,C}\) and \(d_{ijk}^{d,T}=d_{ijk}^{d,C}\) for all
\(i,j=1,2,\ldots,n\) and \(k=1,2,\ldots,m\). In the non-degenerate case
$
\liminf_{n\to\infty}\sqrt n\,\mathrm{se}>0
\quad \text{a.s.},
$
we have
\[
\widehat{\mathrm{GTE}}/\mathrm{se}
\xrightarrow{d} \mathcal N(0,1),
\qquad
\widehat{\mathrm{se}}/\mathrm{se}
\xrightarrow{p}1.
\]
Consequently,
$
\widehat{\mathrm{GTE}}/\widehat{\mathrm{se}}
\xrightarrow{d} \mathcal N(0,1).
$
\end{theorem}

Detailed proofs are deferred to Appendix~\ref{subsec:proof-theorem-aa}.

Algorithm~\ref{alg:gte-se} summarizes the computational procedure for estimating the GTE and conducting inference from share-induced view event logs.

\begin{algorithm}[t]
\caption{Propagation-adjusted estimation and inference for the GTE}
\label{alg:gte-se}
\KwIn{
A collection $\mathcal{L}$ of share-induced view events; assignments $\{(Z_{i},V_{i})\}_{i=1}^{n}$.
}

\BlankLine
\textbf{1. Compute per-user summaries.}

Compute $\{(Y_i^{d},Y_i^{s},W_i^{s})\}_{i=1}^n$ by aggregating share-induced view events according to three distinct counting schemes.

\BlankLine
% \textbf{2. Group sizes and estimates.}
\textbf{2. Estimates.}

% $n_T \leftarrow \sum_{i=1}^n Z_iV_i$;\quad
% $n_C \leftarrow \sum_{i=1}^n (1-Z_i)V_i$
Compute $(\widehat{Y}^{d}_T,\widehat{Y}^{d}_C)$ according to \eqref{eq:arm-specific-mean} and $(\widehat{q}^{T},\widehat{q}^{C})$ according to \eqref{eq:calculate-q}. 

\BlankLine
\textbf{3. Plug-in estimator of GTE.}

$\widehat{\mathrm{GTE}}
\leftarrow
\widehat{Y}_T^{d}/(1-\widehat q^T)
-
\widehat{Y}_C^{d}/(1-\widehat q^C)$

\textbf{4. A/A testing.}

(i) Compute the standard error $\widehat{\mathrm{se}}$ by the formula \eqref{eq:se}.

(ii) Compute the test statistic $\widehat{\mathrm{GTE}}/\widehat{\mathrm{se}}$ and the two-sided $p$-value = $2(1 - \Phi(|\widehat{\mathrm{GTE}}/\widehat{\mathrm{se}}|))$, where $\Phi$ is the CDF of the standard normal distribution.

\KwOut{$\widehat{\mathrm{GTE}},\ \widehat{\mathrm{se}},\ p\text{-value}$}

\end{algorithm}

\section{Extension to Reactivation Rate}

\label{sec:GTE_ra}

While the sharing metrics above quantify differences in expected share-induced engagement, we are also interested in the user reactivation rate driven by social exposure within the observation window, defined as

 \begin{align}
    \mathrm{RA}
    := \frac{1}{n}
    \mathbb{E}\!\left[
        \sum_{i=1}^n \mathbf{1}\left\{
       W_i^s \geq 1 \right\}
    \right],
\end{align}
where the indicator equals one if a user is reactivated by at least one share during the observation window. The corresponding global treatment effect is defined analogously and denoted by $\mathrm{GTE}^{\mathrm{ra}}$.

Estimating the reactivation rate is more challenging because it involves a nonlinear
transformation of the observed data. As a consequence, the flow-balance identity in
equation~(\ref{eq:flow-balance}) is no longer applicable. We instead exploit a Poisson
limit under a regime where the per-edge success probabilities are small.

We assume that for each $(i,j,k)$, the sharing parameters $\{d_{ijk}^{s}, d_{ijk}^{d}\}$
are sufficiently small. In this regime, recall that a sum of independent binomial random
variables with small success probabilities converges to a Poisson random variable. Hence,
it is reasonable to approximate
$
W_i^{s} \sim \mathrm{Poisson}\!\bigl(\lambda_i\bigr).
$
The per-user reactivation probability is therefore approximated by
\[
\mathbb{P}(W_i^{s}\ge 1)\approx 1-\exp\{-\lambda_i\}.
\]

To estimate $\lambda^{T}$ and $\lambda^{C}$ for the treatment and control groups,
respectively, we reuse the quantities obtained in the previous analysis. Under the
Poisson approximation, the per-user expected number of share-driven views under global
treatment and global control is estimated by
\[
\widehat{\lambda}^{a}:=\frac{\widehat{Y}^d_a}{1-\widehat{q}^{a}},\quad a\in\{T,C\}.
\]
This leads to the following heuristic estimator of the global treatment effect on the reactivation rate:
\begin{equation}
\label{eq:GTE-ra-hat}
\widehat{\mathrm{GTE}}^{\mathrm{ra}}
=
\Bigl(1-\exp\{-\widehat{\lambda}^{T}\}\Bigr)
-\Bigl(1-\exp\{-\widehat{\lambda}^{C}\}\Bigr)
=
\exp\{-\widehat{\lambda}^{C}\}-\exp\{-\widehat{\lambda}^{T}\}.
\end{equation}

\section{Numerical Experiments}
\label{sec:numerical}
In Section \ref{subsec:comp-methods}, we introduce the benchmark methods for the metrics of impact-of-sharing and reactivation rates, respectively. Section~\ref{subsec:sim-setup} describes the simulation setup, and Section~\ref{subsec:results} reports the results. Appendix \ref{sec:appendix-more-exps} presents additional simulation results. 
The code for the numerical experiments is publicly available online.\footnote{\url{https://github.com/cwt2001/causal-estimation-with-flywheel}}

\subsection{Benchmark Methods}
\label{subsec:comp-methods}

For the impact-of-sharing metric (IS), we consider the following benchmark methods.
\begin{itemize}
    \item \textbf{Difference-in-means estimator (DM):} ignores network propagation effects, i.e., 
    \[ \widehat{\mathrm{GTE}}_{\text{DM}} = \frac{1}{n_{T}} \sum_i Z_iV_{i} S_{i} - \frac{1}{n_{C}} \sum_i (1-Z_i) V_{i} S_{i},
    \]
    where $S_i := Y_i^{d}+Y_i^{s}$ is the total number of share-view events attributed to user $i$'s sharing.
    \item \textbf{Difference-in-means estimator with first-order resharing (DM-FO):} accounts for first-order propagation by replacing $S_i$ with a first-order adjusted outcome $S_i^{(1)}$, i.e., 
    \[
    \widehat{\mathrm{GTE}}_{\text{DM-FO}}
    =
    \frac{1}{n_{T}} \sum_{i=1}^n Z_i V_i S_i^{(1)}
    -
    \frac{1}{n_{C}} \sum_{i=1}^n (1-Z_i) V_i S_i^{(1)},
    \]
    where $S_i^{(1)}$ counts share-view events attributable to $i$'s direct shares and those induced by one subsequent resharing along the cascade triggered by these shares.
    \item \textbf{Graph cluster randomization (GCR):} 
    assigns treatment at the cluster level after partitioning the graph using the Leiden algorithm~\cite{traag2019louvain}. The estimator is \[ \widehat{\mathrm{GTE}}_{\mathrm{GCR}} = \frac{1}{|\mathcal C_T|}\sum_{c\in\mathcal C_T}\bar W_c^s - \frac{1}{|\mathcal C_C|}\sum_{c\in\mathcal C_C}\bar W_c^s, \] where $\mathcal C_T$ and $\mathcal C_C$ are the treated and control clusters, and $\bar W_c^s$ is the average share-induced views in cluster $c$.
\end{itemize}
% graph cluster randomization

Estimating the reactivation rate (RA) is more challenging, even when constructing suitable benchmark methods, because each user may receive content from both treated and control users. Therefore, below we describe a straightforward approach for allocating outcomes to the treatment and control groups.

For each user $i$, we decompose the total number of share-induced views as
$
W_{i}^{s} = W_{i}^{s,T} + W_{i}^{s,C},
$
where $W_{i}^{s,T}$ counts view events attributed to shares from treated senders,
and $W_{i}^{s,C}$ counts those from control senders. Define the empirical assignment
proportions $\widehat{\pi}_T := n_T/n$ and $\widehat{\pi}_C := n_C/n$, and the
reweighted share-induced view counts
\[
\widetilde{W}_{i}^{s,T} = {W_i^{s,T}}/{\widehat{\pi}_{T}},
\qquad
\widetilde{W}_{i}^{s,C} ={W_i^{s,C}}/{\widehat{\pi}_{C}}.
\]
We allocate one unit of reactivation credit to receiver $i$ and split it between treated
and control senders in proportion to their reweighted share-induced exposures.
The resulting treated-minus-control credit difference is
\[
\widehat r_i
:=
\frac{\widetilde W_i^{s,T}-\widetilde W_i^{s,C}}
{\widetilde W_i^{s,T}+\widetilde W_i^{s,C}},
\]
with the convention that $\widehat r_i = 0$ when
$\widetilde W_i^{s,T}+\widetilde W_i^{s,C}=0$.
We then consider the following estimators for $\mathrm{GTE}^{\mathrm{ra}}$.

\begin{itemize}
        \item \textbf{Exposure-weighted estimator (EW).}
        \[
        \widehat{\mathrm{GTE}}^{\mathrm{ra}}_{\mathrm{EW}}
        :=
        \frac{1}{n}\sum_{i=1}^n \widehat r_i.
        \]
        
        \item \textbf{Holdout exposure-weighted estimator (HEW).}
        Leveraging holdout units ($V_i=0$), we define
        \[
        \widehat{\mathrm{GTE}}^{\mathrm{ra}}_{\mathrm{HEW}}
        :=
        \frac{1}{n-n_V}\sum_{i=1}^n (1-V_i)\,\widehat r_i.
        \]
\end{itemize}

\subsection{Simulation Setups}
\label{subsec:sim-setup}

We use a multivariate Hawkes process \cite{embrechts2011multivariate} to simulate the temporal dynamics of user–content interactions on a network. We simulate the impact of the intervention on user sharing behavior by varying the edge–content specific propagation strengths $d_{ijk}^{s,a},d_{ijk}^{d,a}, a \in \{T,C,O\}$, where $O$ denotes the holdout group: users with $V_i=0$. 

We evaluate our estimators on a synthetic user network generated by a Barab\'asi--Albert model \cite{albert2002statistical} with $n=50{,}000$ nodes and average degree 50, yielding a sparse, large-scale graph with a scale-free degree distribution.

Given the network adjacency matrix $G=(g_{ij})$, we simulate user--content interactions with $m$ contents. The discovery (exogenous) intensity is time-homogeneous and factorized as
\[
\mu_{ik}(t)=\frac{a_i^d b_k^d}{mT},\quad t\in[0,T],
\]
where $a_i^d$ captures user activeness and $b_k^d$ captures content popularity. For each experimental condition $a\in\{T,C,O\}$, we parameterize edge--content propagation strengths for discovery-driven and share-driven events by
\[
d_{ijk}^{d,a}
=
g_{ij}\,\phi_i^{d,a}\,\varphi_j^{d,a}\,\theta_k^{d,a}/\bar g,
\quad
d_{ijk}^{s,a}
=
g_{ij}\,\phi_i^{s,a}\,\varphi_j^{s,a}\,\theta_k^{s,a}/\bar g,
\]
where $\bar g := n^{-1}\sum_{i,j} g_{ij}$ is the average degree, and
$\phi_i^{\cdot,a}$, $\varphi_j^{\cdot,a}$, and $\theta_k^{\cdot,a}$ denote sender-, receiver-, and content-specific propagation effects, respectively.

We draw base parameters independently from uniform distributions:
\begin{align*} 
&a_i^{d} \sim U(0,10), \quad b_k^{d} \sim U(0,20), \\
& \phi_i^{d,T} \sim U(0,1), \quad \varphi_j^{d,T} \sim U(0,0.2), \quad \theta_k^{d,T}\sim U(0,1). 
\end{align*}
We set $\varphi_{j}^{s,T}=\varphi_{j}^{d,T}$ and $\theta_{k}^{s,T}=\theta_{k}^{d,T}$, and allow sharing to differ from discovery via
\[
\phi_i^{s,T}=\max(\phi_i^{d,T}+\varepsilon_i,0),\quad \varepsilon_i\sim U(-0.1,0.3).
\]
To introduce heterogeneity across users, we perturb the sender-side effects $\phi$ for the control and holdout groups while keeping all other parameters fixed. Note that this setting violates Assumption \ref{assum:homogeneity}, which further demonstrates the generalizability of our method. Specifically, for $a\in\{C,O\}$,
\[
\phi_i^{d,a}=\max\!\bigl(\phi_i^{d,T}+\delta_i^{a},\,0\bigr),\quad
\phi_i^{s,a}=\max\!\bigl(\phi_i^{s,T}+\delta_i^{a},\,0\bigr),
\]
where $\delta_i^a \sim U(-\Delta,0)$ are drawn independently, and the disturbance parameter $\Delta$ controls the magnitude of treatment effect.

We consider $\Delta\in\{0.3,0.5,0.7\}$. All other parameters are fixed at $m=2000$, $\pi = 0.5$, and $p=0.5$. For each configuration, we generate $500$ Monte Carlo replications. 
For each replication, the ground-truth $\mathrm{GTE}$ and
$\mathrm{GTE}^{\mathrm{ra}}$ are computed by separately simulating the
global treatment and global control regimes and taking the differences
in the corresponding metrics.

\subsection{Results}
\label{subsec:results}

Figure \ref{fig:barabasi-phi-GTE-estimate-mse} summarizes the performance of different GTE estimators as the disturbance parameter $\Delta$ varies. The left panel shows boxplots of the estimates. The dashed horizontal line marks the mean ground-truth GTE across simulations. The right panel reports the corresponding mean squared error (MSE). Overall, the proposed estimator is nearly unbiased and exhibits smaller bias across all settings. It achieves a lower MSE than DM and GCR, and its MSE is comparable to DM-FO. 

Moreover, using the asymptotic variance in Theorem~\ref{theorem:aa-normality},
we report the empirical 95\% coverage of nominal Wald intervals for the ground-truth GTE
in our A/B simulations in Table~\ref{tab:barabasi_coverage_pi05}.
Although the CLT is proved only under A/A tests,
coverage is close to the nominal level across the simulated settings,
supporting the robustness of Theorem~\ref{theorem:aa-normality} for A/B tests.

\begin{figure}[ht]
  \centering
  \begin{minipage}[t]{0.45\columnwidth}
    \centering
    \includegraphics[width=\linewidth]{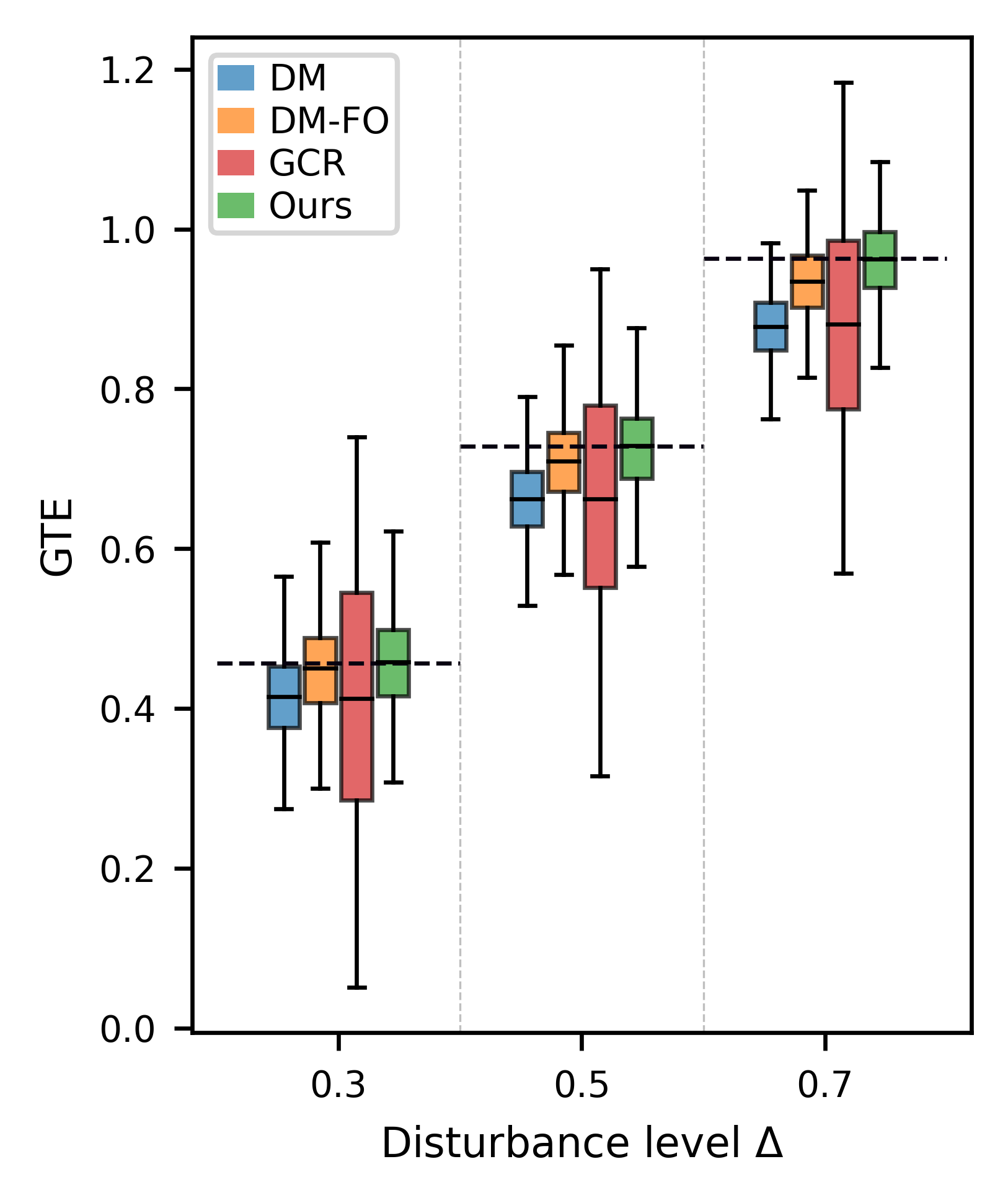}
    {\small (a) Boxplots of GTE Estimates}
  \end{minipage}\hfill
  \begin{minipage}[t]{0.45\columnwidth}
    \centering
    \includegraphics[width=\linewidth]{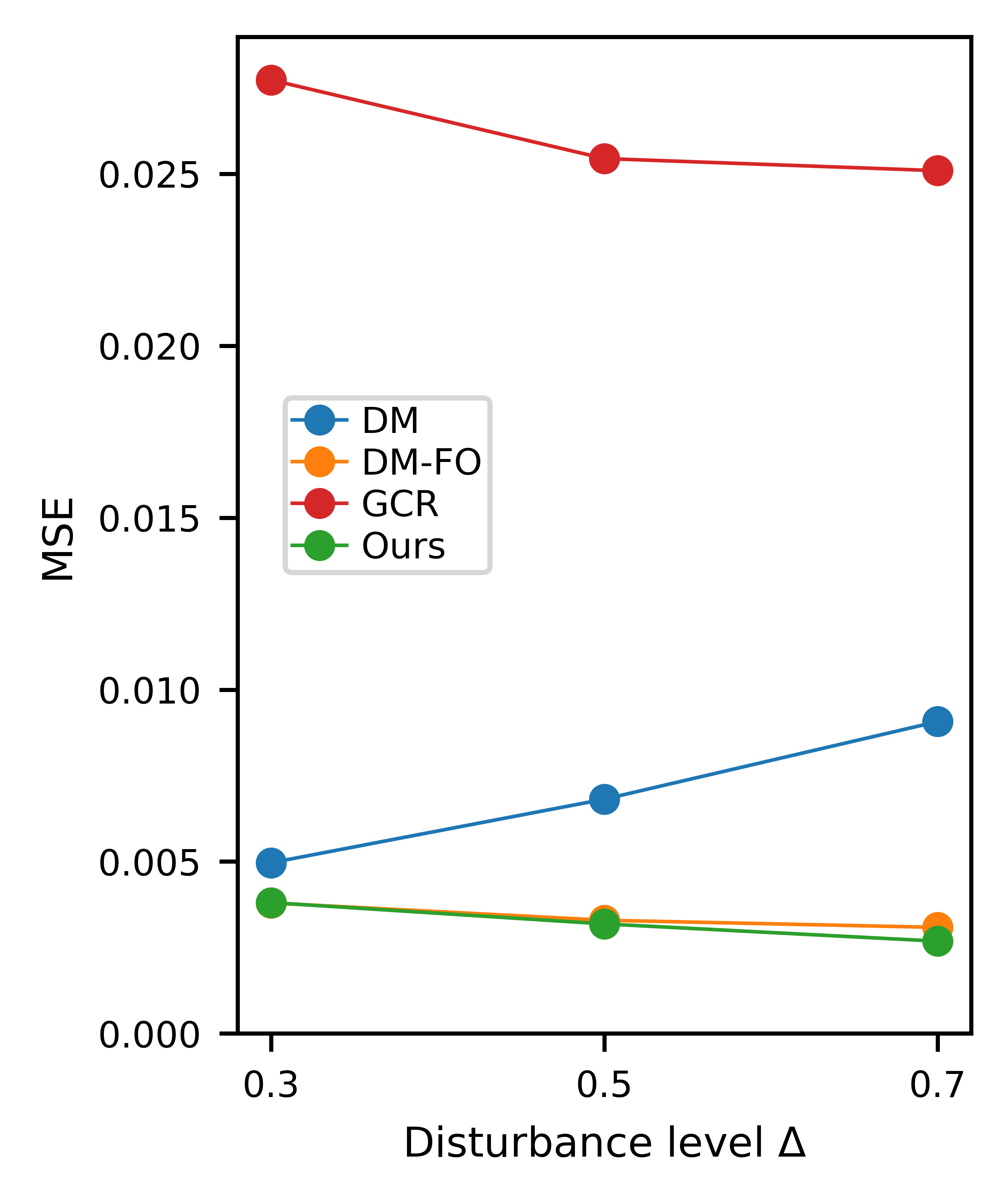}
    {\small (b) Mean squared error (MSE)}
  \end{minipage}
  \caption{Performance of different estimators for the impact of sharing under varying disturbance levels.}
  \label{fig:barabasi-phi-GTE-estimate-mse}
\end{figure}

\begin{table}[t]
\centering
\caption{Empirical 95\% coverage under $\pi=0.5$ and different disturbance parameters $\Delta$.}
\label{tab:barabasi_coverage_pi05}
\begin{tabular}{c ccc}
\toprule
$\Delta$ & 0.3 & 0.5 & 0.7 \\
\midrule
95\% coverage & 0.944 & 0.948 & 0.958 \\
\bottomrule
\end{tabular}

\end{table}

Figure~\ref{fig:barabasi-phi-GTE-ra-estimate-mse} reports the corresponding results for $\mathrm{GTE}^{\mathrm{ra}}$. Although the proposed estimator is visibly biased, its bias remains smaller than both EW and HEW, which results in a lower MSE than the two competing methods. Note that the MSE curves for EW and HEW overlap.

\begin{figure}[ht]
  \centering
  \begin{minipage}[t]{0.45\columnwidth}
    \centering
    \includegraphics[width=\linewidth]{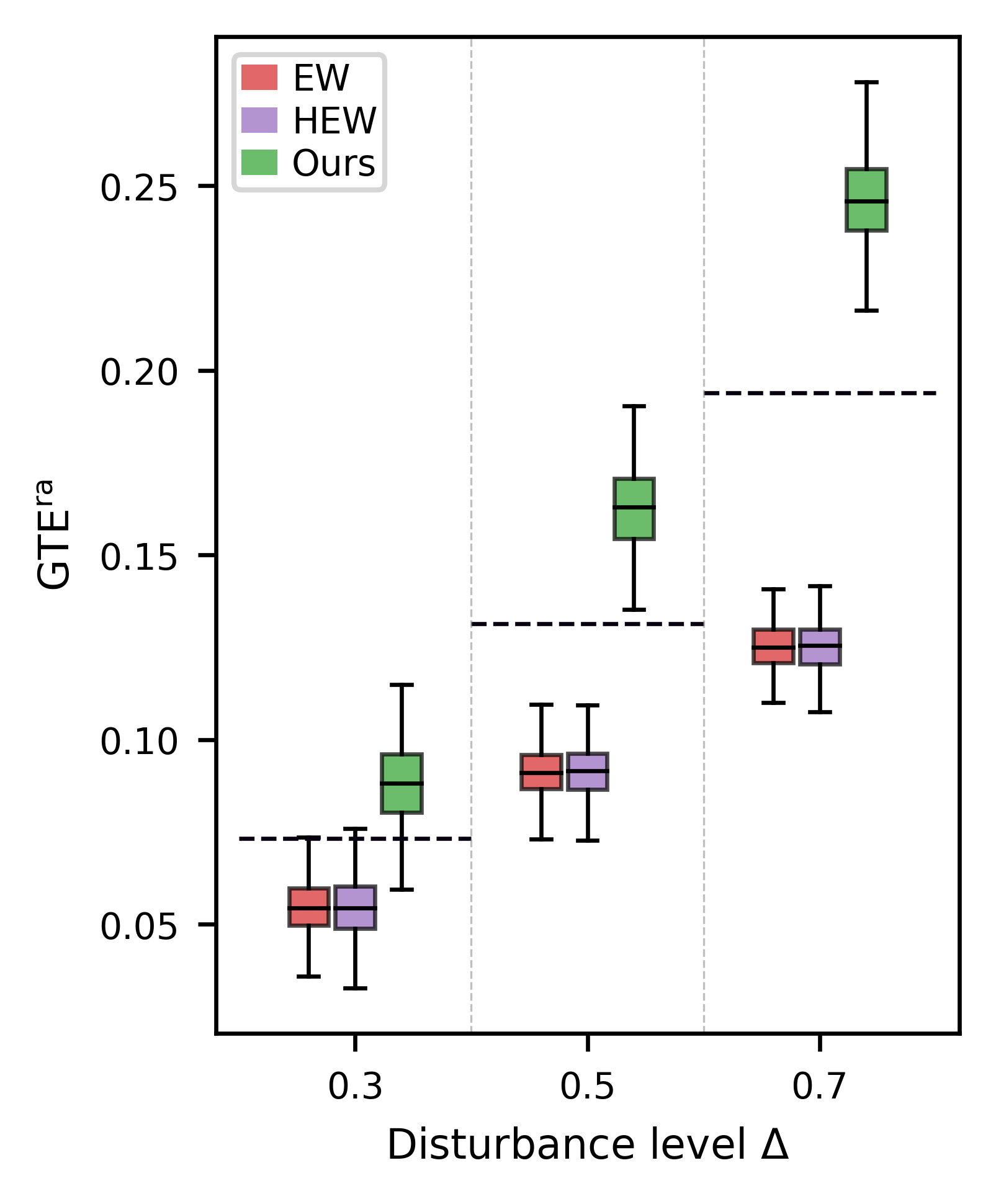}
    {\small (a) Boxplots of $\text{GTE}^{\text{ra}}$ Estimates}
  \end{minipage}\hfill
  \begin{minipage}[t]{0.45\columnwidth}
    \centering
    \includegraphics[width=\linewidth]{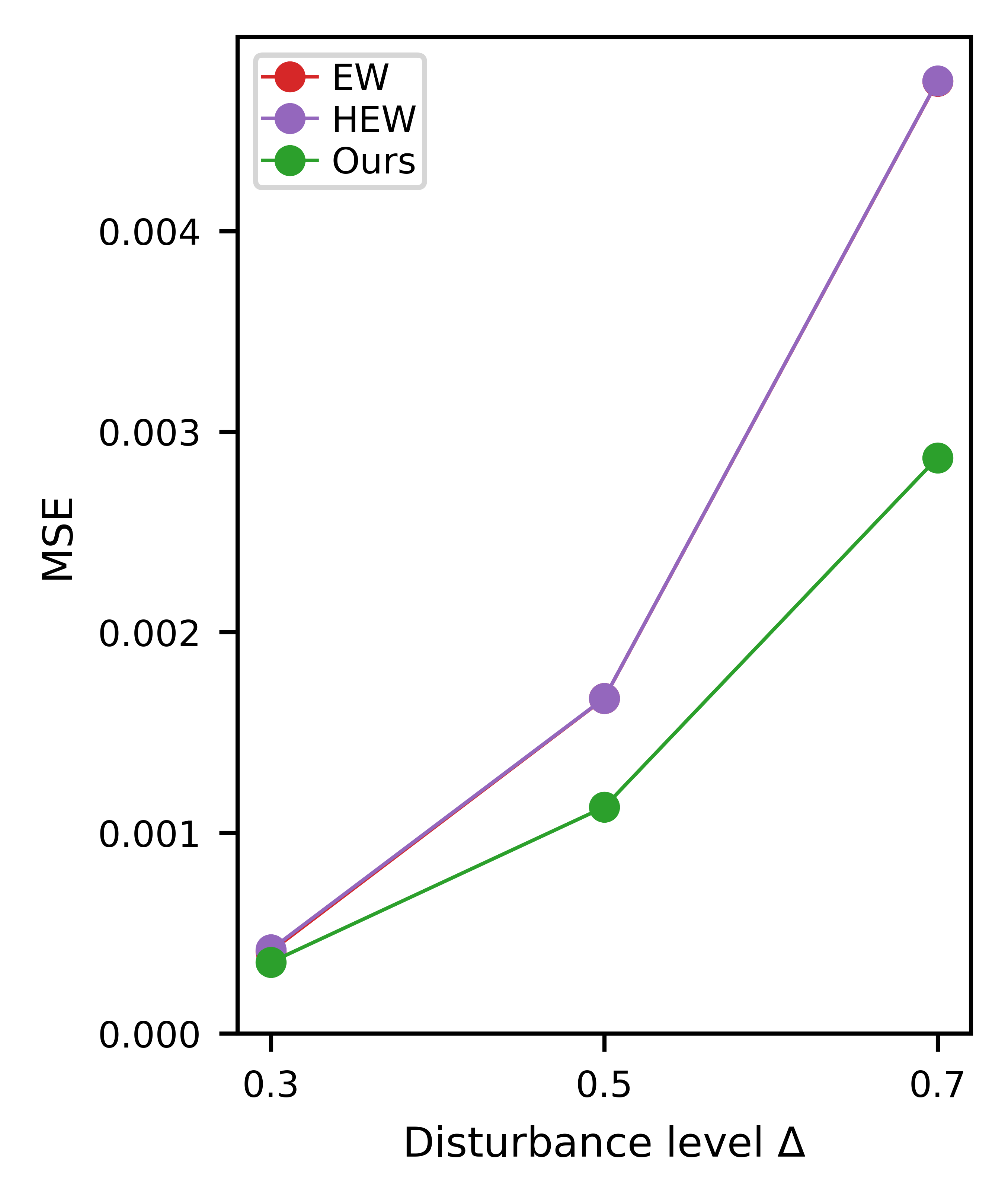}
    {\small (b) Mean squared error (MSE)}
  \end{minipage}
  \caption{Performance of different estimators for the reactivation rate  under varying disturbance levels.}
  \label{fig:barabasi-phi-GTE-ra-estimate-mse}
\end{figure}

\FloatBarrier

\section{Real-World Applications}
\label{sec:realworld}
We partnered with a major social networking platform with millions to billions of users
to evaluate the performance of our proposed methods in a real-world setting. The platform
continuously develops and deploys new sharing features, making reliable evaluation under
interference essential. 
Prior to implementing our methods, we followed the approach of \citet{han2023detect}
to test for the presence of interference. Empirical evidence from extensive experiments
shows that engagement-related metrics often exhibit amplification patterns as treatment
exposure increases, which is consistent with interference induced by social sharing and
motivates the development of our methods.

Our implementation leverages commonly available attribution logs to enable estimation of
the global treatment effect while accounting for multi-round diffusion. It is compatible
with Bernoulli randomization and modern online analytical processing engines, and can be
deployed automatically at scale.

We validate our approach using real data on a sharing-related metric from one channel,
a built-in feature that allows users to discover, watch, and share videos and live
streams. We first conduct an A/A test in which the treatment and control are identical.
We collect data from all users over one week and randomly assign 50\% of users to each
group. We then compute the proposed estimator and the corresponding $p$-values.

After 200 replicates, we examine whether the $p$-values approximately follow a uniform
distribution under the A/A test. Figure~\ref{fig:aa} shows that the $p$-values from the
proposed method are close to uniform. The remaining deviation is mainly due to the
heavy-tailed nature of sharing behavior, where a small number of users share a very large
amount of content with others.

\begin{figure}[h]
    \centering
    \includegraphics[width=0.55\linewidth, height=0.55\linewidth]{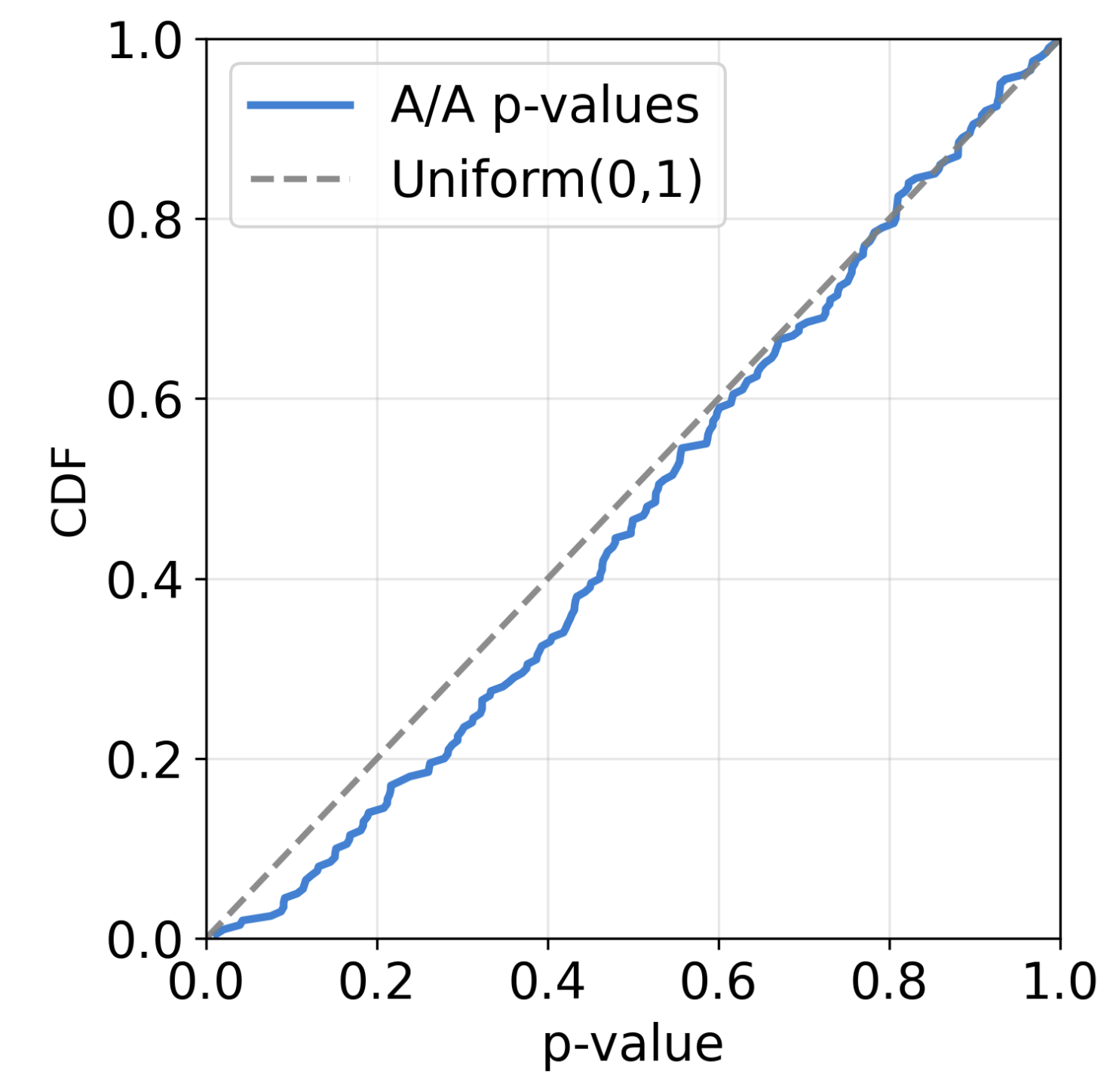}
    % \caption{$p$-value's distribution for 200 A/A experiments.}
    \caption{Distribution of \(p\)-values for 200 A/A experiments.}
    \label{fig:aa}
\end{figure}

We further demonstrate the practical value of our approach through a real A/B experiment.
In this experiment, the treatment group was exposed to a redesigned interface intended to
encourage sharing and downstream engagement with shared content, while the control group
retained the original design. We focus on the impact-of-sharing metrics (Eq. \ref{def:IS}). By construction, such sharing features can reactivate users
and propagate engagement through social connections over multiple rounds, making
interference unavoidable. 

Table~\ref{tab:abreal} shows that the conventional difference-in-means estimator and
first-order propagation metrics underestimate the true impact and yield statistically
insignificant results, whereas our estimator produces a statistically significant effect.
Based on our estimator’s results, the platform decided to launch the feature to all users,
and a comparison between pre-launch and post-launch outcomes further validates the
effectiveness of the feature. 
We also report a companion A/A experiment in Table~\ref{tab:abreal}, where the control
group is randomly split into two subgroups (A1 and A2), to further verify the validity
of the randomization. Together, these results highlight the importance of explicitly
accounting for flywheel-driven interference when evaluating sharing-feature designs on platforms.

\FloatBarrier

\begin{table}[htbp]
\caption{ A/A and A/B testing results on the platform.} 
% \vspace{-0.1in}
\begin{center}
{
\begin{tabular}{cccccccc}
  \toprule
  & Method & Evaluation & A/A & A/B \\
  \midrule
  & \multirow{2}*{Ours} &relative difference & 0.06\% & 0.548\% \\
                       & & $\mathsf{p\text{-}value}$  & 0.533 &
                         \textbf{0.005}  \\
                           \midrule
                       & \multirow{2}*{DM} & relative difference & 0.096\% & 0.293\% \\  % 62416 
                       & & $\mathsf{p\text{-}value}$  &  0.733  & 0.123  \\
                       \midrule
                       & \multirow{2}*{DM-FO} & relative difference & 0.035\% & 0.267\% \\ % 68876
                       & & $\mathsf{p\text{-}value}$  & 0.813 & 0.160 \\
  \bottomrule
\end{tabular}
}
\end{center}
\label{tab:abreal}
\end{table}

%\section{Conclusion}

\FloatBarrier

\begin{acks}
This work is funded by the Young Scientists Fund - Category C (Grant No. 72501243) from the Natural Science Foundation of China (NSFC) and Tencent Weixin Rhino-Bird Focused Research Program (Grant No. WXG-FR-2025-09).
\end{acks}

\phantomsection
\addcontentsline{toc}{section}{References}
\bibliographystyle{unsrtnat}
\bibliography{references}

\clearpage
\appendix
\section{Proofs}

\begin{proof}[Proof of Proposition \ref{prop:IS-finite}]
\label{subsec:proof-of-IS-finite}
Under Assumption \ref{assum:bounded}(2), \(\|Q^{(k)}\|_\infty\le \bar{c}<1\) for every \(k\). Hence a single discovery-driven view event of content \(k\) triggers, in expectation, at most
\[
\sum_{\ell\ge1}\|Q^{(k)}\|_\infty^\ell \le \frac{\bar{c}}{1-\bar{c}}
\]
share-induced view events over all later generations. Therefore
\[
\mathbb E\!\left[\sum_{i,k}N^s_{ik}(\infty)\right]
\le \frac{\bar{c}}{1-\bar{c}}\sum_{i,k}\mathbb E[N^d_{ik}(\infty)]
\le \frac{\bar{c}}{1-\bar{c}}\,n B,
\]
where the last inequality follows from Assumption \ref{assum:bounded}(1). Since \(O_{ik}=1\), this gives
\[
\mathrm{IS}=n^{-1}\mathbb E\!\left[\sum_{i,k}N^s_{ik}(\infty)\right]\le \frac{\bar{c} B}{1-\bar{c}} <\infty.
\]
\end{proof}

\begin{proof}[Proof of Proposition \ref{prop:refined-flow-balance-param-form}]
\label{subsec:proof-of-flow-balance-param-form}

Taking expectations in the flow-balance identity~\eqref{eq:flow-balance2}, and using the Hawkes cluster representation, we note that conditional on a parent event of type \(v\in\{d,s\}\) occurring at \((i,k)\), the number of its direct offspring sent to receiver \(j\) has conditional mean \(d^v_{ijk}\). Moreover, this offspring generation is independent across parent events.

Hence, by Tonelli's theorem,
\[
\mathbb E\!\left[\int_0^\infty D^v_{ijk}(t)\,dN^v_{ik}(t)\right]
=d^v_{ijk}\,\mathbb E[N^v_{ik}(\infty)].
\]
For \(v=d\), this equals \(w^d_{ik}d^d_{ijk}\); for \(v=s\), it equals \(w^s_{ik}d^s_{ijk}\). Substituting these two identities into the expectation of (\ref{eq:flow-balance2}) and using
\(q_{ik}=\sum_j d^s_{ijk}\) gives
\[
\sum_{i,j,k}w^d_{ik}d^d_{ijk}
=\sum_{i,k}w^s_{ik}(1-q_{ik}),
\]
as claimed. 
\end{proof}

\begin{proof}[Proof of Theorem \ref{theorem:consistency}]
\label{subsec:proof-theorem-consistency}
% Proof sketch
 The proof is outlined as follows: We first express the GTE in terms of discovery-origin views and the effective downstream sharing rate. Then, we show that the sample averages of discovery-origin counts and the plug-in downstream rates converge in probability to their expectations. Finally, the continuous mapping theorem gives the consistency of the plug-in GTE estimator.
 
For \(a\in\{T,C\}\), define
\[
\mu_a^d
:=
\frac1n\sum_{i=1}^n\sum_{j=1}^n\sum_{k=1}^m
w_{ik}^d d_{ijk}^{d,a}.
\]
Under Assumption~\ref{assum:homogeneity}, \(q_{ik}^a\equiv q^a\). Hence,
by Proposition~\ref{prop:refined-flow-balance-param-form},
\[
\mathrm{IS}^a=\frac{\mu_a^d}{1-q^a},
\qquad
\mathrm{GTE}
=
\frac{\mu_T^d}{1-q^T}
-
\frac{\mu_C^d}{1-q^C}.
\]
Notice that
\[
\frac{\widehat Y_a^d}{1-\widehat q^a}
-
\frac{\mu_a^d}{1-q^a}
=
\frac{\widehat Y_a^d-\mu_a^d}{1-q^a}
+
\frac{\widehat Y_a^d(\widehat q^a-q^a)}
{(1-\widehat q^a)(1-q^a)}.
\]
Since \(q^a\le \bar c<1\), it is enough to show, for each
\(a\in\{T,C\}\), that
\[
\widehat Y_a^d-\mu_a^d=o_p(1),
\qquad
\widehat Y_a^d=O_p(1),
\qquad
\widehat q^a-q^a=o_p(1).
\]

We only consider the treatment case. The control case is identical. Let
$
\widetilde Y_T^d
=
n_T/(n\pi_T)\widehat Y_T^d .
$
Because the intervention changes sender-side propagation only, we have \(\mathbb E[\widetilde Y_T^d]=\mu_T^d\).

We next bound the variance. Let \(W_i^d:=\sum_{k=1}^mN_{ik}^d(\infty)\).
By Assumption~\ref{assum:bounded}(1), \(W_i^d\) is Poisson with mean at most
\(B\). Conditional on the discovery-origin views generated by user \(i\),
the number of direct share-induced views generated from these views has
conditional mean at most
\[
\sum_{k=1}^mN_{ik}^d(\infty)\sum_{j=1}^n d_{ijk}^{d,T}
\le
\bar c\, W_i^d .
\]
Therefore,
$
\sup_i\mathbb E[(Y_i^d)^2]\le B^2+2B<\infty,
$
and the independence of discovery-origin views and their direct offspring
across senders gives
\[
\operatorname{Var}(\widetilde Y_T^d)
\le
\frac{1}{n^2\pi_T^2}
\sum_{i=1}^n \mathbb E[(Y_i^d)^2]
=O(n^{-1}).
\]
Thus \(\widetilde Y_T^d-\mu_T^d=o_p(1)\). Moreover, 
$
\mu_T^d
\le
\bar c B,
$
so \(\widetilde Y_T^d=O_p(1)\). Since
\(n_T/(n\pi_T)\xrightarrow{p}1\),
\[
\widehat Y_T^d-\mu_T^d
=
\left(\frac{n \pi_{T}}{n_T}-1\right)\widetilde Y_T^d
+
(\widetilde Y_T^d-\mu_T^d)
=o_p(1).
\]
Besides, \(\widehat Y_T^d=O_p(1)\).

It remains to prove \(\widehat q^T-q^T=o_p(1)\). Write
$
F_{n,T}:=\sum_{i=1}^n Z_iV_i W_i^s,$ 
$
G_{n,T}:=\sum_{i=1}^n Z_iV_i Y_i^s .
$
Then \(\widehat q^T=G_{n,T}/F_{n,T}\) on \(\{F_{n,T}>0\}\), with the
convention \(\widehat q^T=0\) when \(F_{n,T}=0\).

Define
\[
H_{n,T}
:=
\sum_{i=1}^n\sum_{j=1}^n\sum_{k=1}^m
Z_iV_i Z_jV_j
\int_0^\infty D_{ijk}^d(t)\,dN_{ik}^d(t).
\]
Every event counted by \(H_{n,T}\) is a share-induced view received by a
treated user. Therefore,
$
H_{n,T}\le F_{n,T}.
$
Conditional on the assignment and the discovery-origin view processes,
\(H_{n,T}\) is Poisson with mean
\[
\Lambda_{n,T}
:=
\sum_{i=1}^n\sum_{j=1}^n\sum_{k=1}^m
Z_iV_i Z_jV_j N_{ik}^d(\infty)d_{ijk}^{d,T}.
\]
Since conditional on the assignment,
\[
\begin{aligned}
&\operatorname{Var}\!\left(
\sum_{i,j,k}Z_iV_i Z_jV_j
\{N_{ik}^d(\infty)-w_{ik}^d\}d_{ijk}^{d,T}
\,\middle|\,\{Z_i,V_{i}\}_{i=1}^n
\right)  \\
&\qquad\le
\sum_{i,k}w_{ik}^d
\left(\sum_j d_{ijk}^{d,T}\right)^2
\le
\bar c^2\sum_{i,k}w_{ik}^d
\le
\bar c^2Bn .
\end{aligned}
\]
we have
\[
\frac{\Lambda_{n,T}}{n}
-
\frac1n\sum_{i,j,k}Z_iV_i Z_jV_j w_{ik}^d d_{ijk}^{d,T}
=o_p(1).
\]

Let
$
\mu_{ij,T}^d:=\sum_{k=1}^m w_{ik}^d d_{ijk}^{d,T}.
$
Since \(d_{iik}^{d,T}=0\), \(\mu_{ii,T}^d = 0\). Moreover,
\[
\max_i\sum_j \mu_{ij,T}^d\le \bar c B,
\qquad
\max_j\sum_i \mu_{ij,T}^d\le B.
\]
The second inequality is exactly the bound in
Assumption~\ref{assum:bounded}(1), while the first follows from Assumptions~\ref{assum:bounded}(1)--(2). A covariance expansion for the
Bernoulli quadratic form gives
\[
\begin{aligned}
& \operatorname{Var}\!\left(\sum_{i,j}Z_iV_i Z_jV_j \mu_{ij,T}^d\right)\\
& \le
C_\pi\left\{
\sum_{i,j}(\mu_{ij,T}^d)^2
+
\sum_i\left(\sum_j\mu_{ij,T}^d\right)^2
+
\sum_j\left(\sum_i \mu_{ij,T}^d\right)^2
\right\} =O(n),
\end{aligned}
\]
where \(C_\pi<\infty\) depends only on \(\pi_T\). Therefore,
\[
\frac1n\sum_{i,j}Z_iV_{i} Z_{j}V_{j} \mu_{ij,T}^d
-
\pi_T^2\frac1n\sum_{i,j}\mu_{ij,T}^d
=o_p(1).
\]
By Assumption~\ref{assum:bounded}(3),
\[
\frac1n\sum_{i,j}\mu_{ij}^T
=
\frac1n\sum_{i,j,k}w_{ik}^d d_{ijk}^{d,T}
\ge \underline c .
\]
Thus
\[
\mathbb P\!\left(
\frac{\Lambda_{n,T}}{n}
\ge
\frac{\pi_T^2\underline c}{2}
\right)\to1.
\]
Since \(H_{n,T}\mid \Lambda_{n,T}\sim\operatorname{Poisson}(\Lambda_{n,T})\)
and \(\mathbb E[\Lambda_{n,T}]=O(n)\), we have
\[
\frac{H_{n,T}-\Lambda_{n,T}}{n}=o_p(1).
\]
Consequently,
\[
\mathbb P\!\left(
F_{n,T}\ge \frac{\pi_T^2\underline c}{4}n
\right)
\ge
\mathbb P\!\left(
H_{n,T}\ge \frac{\pi_T^2\underline c}{4}n
\right)
\to1 .
\]

Under the cluster representation of the Hawkes process, enumerate all
share-induced views in generation order. Let \((i_r,k_r)\),
\(r=1,\ldots,\tau_n\), denote the user--content pair associated with the
\(r\)-th share-induced view, where \(\tau_n\) is the total number of
share-induced views. Let \(L_r\) be the number of direct offspring
share-induced views generated by this \(r\)-th share-induced view. Then
$
F_{n,T}=\sum_{r=1}^{\tau_n} Z_{i_r}V_{i_r},$ 
$
G_{n,T}=\sum_{r=1}^{\tau_n} Z_{i_r}V_{i_r}L_r .
$

Let \(\mathcal G_{n,r-1}\) be the \(\sigma\)-field containing the experimental
assignments, all discovery-origin views, and all share-induced views up to
the \(r\)-th share-induced view, but not the offspring generated by the
\(r\)-th share-induced view. On the event \(\{Z_{i_r}V_{i_r}=1\}\), the
sender is treated. Hence Assumption~\ref{assum:homogeneity} and the Hawkes
cluster construction imply
$
\mathbb E[L_r\mid \mathcal G_{n,r-1}]=q^T, \ 
\operatorname{Var}(L_r\mid \mathcal G_{n,r-1})=q^T .
$
Therefore,
$
D_{n,r}:=\mathbf 1\{r\le \tau_n\}Z_{i_r}V_{i_r}(L_r-q^T)
$
is a martingale-difference array. Define
\[
M_{n,T}:=\sum_{r\ge1}D_{n,r}=G_{n,T}-q^T F_{n, T} .
\]
Moreover,
\[
\mathbb E[M_{n,T}^2]
=
\sum_{r\ge1}\mathbb E[D_{n,r}^2]
\le
q^T \sum_{r\ge1}\mathbb E[\mathbf 1\{r\le \tau_n\}Z_{i_r}V_{i_r}]
\le Cn,
\]
where $C$ is a constant independent of \(n\) and the last inequality follows from Proposition~\ref{prop:IS-finite}.

Let \(\eta_T:=\pi_T^2\underline c/4\). For any \(\varepsilon>0\),
\[
\begin{aligned}
\mathbb P(|\widehat q^T-q^T|>\varepsilon)
&\le
\mathbb P(F_{n,T}<\eta_T n)
+
\mathbb P(|M_{n,T}|>\varepsilon\eta_T n)\\
&\le
o(1)
+
\frac{\mathbb E[M_{n,T}^2]}
{\varepsilon^2\eta_T^2n^2}
=o(1).
\end{aligned}
\]
Thus \(\widehat q^T-q^T=o_p(1)\).

The proof for the control setting is similar. Thus, $\widehat{\text{GTE}} - \text{GTE} \overset{p}{\to} 0$.
\end{proof}

\begin{proof}[Proof of Theorem~\ref{theorem:aa-normality}]
\label{subsec:proof-theorem-aa}

% Proof sketch for theorem 5.1
The proof is outlined as follows: using the Hawkes cluster representation together with Assumption \ref{assum:bounded}, we obtain uniform fourth-moment bounds for $U_i$, ensuring that the remainder terms in the mean-value expansion of $g(U)$ are negligible. A finite-population central limit theorem for complete randomization then implies asymptotic normality of the linear term, and Slutsky's theorem yields the asymptotic normality of the estimator with consistent plug-in standard error.

Under A/A testing, 
$
\bm U_i=(Y_i^d,Y_i^s,W_i^s)^\top
$
does not depend on \(Z_i\).  We first record a moment implication of Assumption~\ref{assum:bounded}. By the Hawkes cluster representation \citep{hawkes1974cluster} and standard moment bounds for subcritical branching processes \citep{athreya2012branching}, the conditions in
Assumption~\ref{assum:bounded}(1)--(2) imply that there exists a constant
\(C<\infty\), independent of \(n\), such that
\[
\sup_{1\le i\le n}\mathbb E\|\bm U_i\|^4\le C .
\]
Consequently,
\[
n^{-1} \sum_{i=1}^n \|\bm U_i\|^4=O_p(1),
\qquad
\max_{1\le i\le n}\|\bm U_i\|=o_p(n^{1/2}).
\]

Recall that
$
\widehat{\bm U}_T
=
(\widehat Y_T^d,\widehat Y_T^s,\widehat W_T^s)^\top,
\widehat{\bm U}_C
=
(\widehat Y_C^d,\widehat Y_C^s,\widehat W_C^s)^\top,
$
and
$
\widehat{\mathrm{GTE}}
=
g(\widehat{\bm U}_T)-g(\widehat{\bm U}_C),
g(u_1,u_2,u_3)={u_1}/{(1-u_2/u_3)}.
$
% Let \(\gamma_V=\nabla g(\bar{\bm U}_V)\).

We next show that \(\nabla g\) is evaluated away from its singularity. By the flow-balance identity, after summing over all users,
$
\sum_{i=1}^n (W_i^s-Y_i^s)=\sum_{i=1}^n Y_i^d .
$
Moreover,
$
\mathbb E[Y_i^d]
=
\sum_{k=1}^m\sum_{j=1}^n w_{ik}^d d_{ijk}^d .
$
The variables \(\{Y_i^d\}_{i=1}^n\) are independent across users and have
uniformly bounded second moments. Hence, by
Assumption~\ref{assum:bounded}(3),
\[
\mathbb P\!\left(
\frac1n\sum_{i=1}^n(W_i^s-Y_i^s)\ge \frac{\underline c}{2}
\right)
=
\mathbb P\!\left(
\frac1n\sum_{i=1}^nY_i^d\ge \frac{\underline c}{2}
\right)
\to 1 .
\]
Conditional on \(\{W_i^s-Y_i^s\}_{i=1}^n\),
\[
\mathbb E\!\left[
\left.
\left\{
\frac1n\sum_{i=1}^n(V_i-\pi)(W_i^s-Y_i^s)
\right\}^2
\right|
\{W_i^s-Y_i^s\}_{i=1}^n
\right]
\le
\frac{C}{n^2}\sum_{i=1}^n(W_i^s-Y_i^s)^2 .
\]
The fourth-moment bound above implies
\(n^{-1}\sum_{i=1}^n(W_i^s-Y_i^s)^2=O_p(1)\). Therefore,
\[
\frac1n\sum_{i=1}^n V_i(W_i^s-Y_i^s)
=
\pi\frac1n\sum_{i=1}^n(W_i^s-Y_i^s)+o_p(1).
\]
Together with \(n_V/n\xrightarrow{p}\pi\), we have
$
\mathbb P\!\left(
\frac1{n_V}\sum_{i=1}^n V_i(W_i^s-Y_i^s)
\ge \frac{\underline c}{4}
\right)\to 1 .
$
That is, $\mathbb P\!\left(
(\bar{\bm U}_{V})_{3}-(\bar{\bm U}_{V})_{2}
\ge \frac{\underline c}{4}
\right)\to 1, $ where \((\bm u)_\ell\) denotes the \(\ell\)-th component of a vector \(\bm u\).
The same argument yields analogous lower bounds for the treatment and control averages.

By a mean-value expansion of $g$ around $\bar{\bm U}_V$, there exist intermediate points
$\widetilde{\bm U}_a$, $a \in \{T,C\}$ on the line segment between $\widehat{\bm U}_a$ and $\bar{\bm U}_V$ such that
\begin{equation}
\begin{aligned}
    \label{eq:lin-g}
    & g(\widehat{\bm U}_a)-g(\bar{\bm U}_V)
    =
    \nabla g(\widetilde{\bm U}_a)^\top(\widehat{\bm U}_a-\bar{\bm U}_V).
\end{aligned}
\end{equation}
Hence,
\begin{equation}
\begin{aligned}
\label{eq:lin-h}
\widehat{\mathrm{GTE}}
& = \gamma_{V}^\top \widehat{\bm U}_T 
- \gamma_{V}^\top \widehat{\bm U}_C
+ R_n,
\end{aligned}
\end{equation}
where the remainder term is
\[
R_n
:=
\bigl(\nabla g(\widetilde{\bm U}_T)-\gamma_{V}\bigr)^\top(\widehat{\bm U}_T-\bar{\bm U}_V)
-
\bigl(\nabla g(\widetilde{\bm U}_C)-\gamma_{V}\bigr)^\top(\widehat{\bm U}_C-\bar{\bm U}_V).
\]
Since \(n_V/n\to_p\pi>0\), \(n_T/n\to_p\pi p>0\), and
\(n_C/n\to_p\pi(1-p)>0\), all sample sizes are of order \(n\) with
probability tending to one. Conditional on 
\[
\mathcal H_n
:=
\sigma\bigl(\{V_i\}_{i=1}^n,n_T,n_C,\{\bm U_i\}_{i=1}^n\bigr),
\]
the treatment assignment is complete randomization over the \(n_V=n_T+n_C\)
exposed users \cite{ding2024first}.
By the randomization variance bound and
\(n^{-1}\sum_{i=1}^n\|\bm U_i\|^2=O_p(1)\), we have
$
\|\widehat{\bm U}_T-\bar{\bm U}_V\|=O_p(n^{-1/2}), 
\|\widehat{\bm U}_C-\bar{\bm U}_V\|=O_p(n^{-1/2}).
$
Since the denominators in \(\nabla g\) are
bounded away from zero with probability tending to one,
\[
\|\nabla g(\widetilde{\bm U}_T)-\gamma_V\|=o_p(1),
\qquad
\|\nabla g(\widetilde{\bm U}_C)-\gamma_V\|=o_p(1).
\]
Thus \(R_n=o_p(n^{-1/2})\). The non-degeneracy condition implies 
\(R_n=o_p({\mathrm{se}})\).

It remains to establish the central limit theorem for the linear term. 
Previous results imply \(\|\gamma_V\|=O_p(1)\). Moreover, since
$
\mathrm{se}^2
=
{n_V}{(n_Tn_C)^{-1}}\gamma_V^\top\Sigma_V\gamma_V,
$
the non-degeneracy condition and
\(n_T/n\to\pi p\), \(n_C/n\to\pi(1-p)\), \(n_V/n\to\pi\) imply that
\(\gamma_V^\top\Sigma_V\gamma_V\) is bounded away from zero with probability
tending to one. Hence,
\[
\frac{
\max_{i:V_i=1}
\left\{\gamma_V^\top(\bm U_i-\bar{\bm U}_V)\right\}^2
}{
\min(n_T,n_C)\gamma_V^\top\Sigma_V\gamma_V
}
=o_p(1).
\]
The finite-population central limit theorem for complete randomization then gives,
conditionally on \(\mathcal H_n\),
\[
\frac{
\gamma_V^\top\widehat{\bm U}_T
-
\gamma_V^\top\widehat{\bm U}_C
}{
\mathrm{se}
}
\xrightarrow{d}
\mathcal N(0,1).
\]
Combining this display with \(R_n=o_p(\mathrm{se})\) and Slutsky's theorem yields
\[
\frac{\widehat{\mathrm{GTE}}}{\mathrm{se}}
=
\frac{
\gamma_V^\top\widehat{\bm U}_T
-
\gamma_V^\top\widehat{\bm U}_C
}{
\mathrm{se}
}
+
\frac{R_n}{\mathrm{se}}
\xrightarrow{d}
\mathcal N(0,1).
\]

Finally, the same finite-population law of large numbers gives
$
\widehat{\Sigma}_T-\Sigma_V=o_p(1),
\widehat{\Sigma}_C-\Sigma_V=o_p(1),
$
under the randomization distribution~\cite{li2017general}.
Therefore,
\[
\widehat{\mathrm{se}}^{\,2}
=
\left(\frac1{n_T}+\frac1{n_C}\right)
\gamma_V^\top\Sigma_V\gamma_V
\{1+o_p(1)\}
=
\mathrm{se}^2\{1+o_p(1)\}.
\]
Thus,
$
{\widehat{\mathrm{se}}}/{\mathrm{se}}\xrightarrow{p}1.
$
Another application of Slutsky's theorem gives
$
{\widehat{\mathrm{GTE}}}/{\widehat{\mathrm{se}}}
\xrightarrow{d}
\mathcal N(0,1).
$
Since the conditional limiting distribution does not depend on \(\mathcal H_n\),
the same convergence also holds unconditionally.
\end{proof}

\section{Additional Simulation Studies}
\label{sec:appendix-more-exps}

This section reports additional simulation results under a range of settings to assess the robustness and sensitivity of our method. Unless otherwise stated, we use the setting from Section~\ref{subsec:sim-setup} as the default setting: a Barab\'asi--Albert graph with average degree $50$, $n=50{,}000$ users, $m=2{,}000$ contents, treatment probability $p=0.5$, exposure probability $\pi=0.5$, the default parameter distribution, and disturbance level $\Delta=0.5$. For one-dimensional parameter sweeps, all non-swept parameters are fixed at their default values.
Because the appendix considers many additional settings, we use 100 Monte Carlo replications per setting and all reported MSEs are averages over these replications.

\subsection{Sensitivity to Exposure Probability}
\label{subsec:appendix-pi-sweep}

We first vary the exposure probability $\pi\in\{0.2,0.5,0.8\}$. This experiment examines the sensitivity of our method to the fraction of exposed users. Figure~\ref{fig:barabasi-pi-GTE-estimate-mse} shows that the proposed estimator remains close to the true GTE and achieves the lowest or nearly lowest MSE across all settings. The improvement is clear compared with DM and GCR. Compared with DM-FO, the proposed estimator is competitive in bias and slightly better in MSE.

\begin{figure}[!htbp]
  \centering
  \begin{minipage}[t]{0.49\textwidth}
    \centering
    \includegraphics[width=\linewidth]{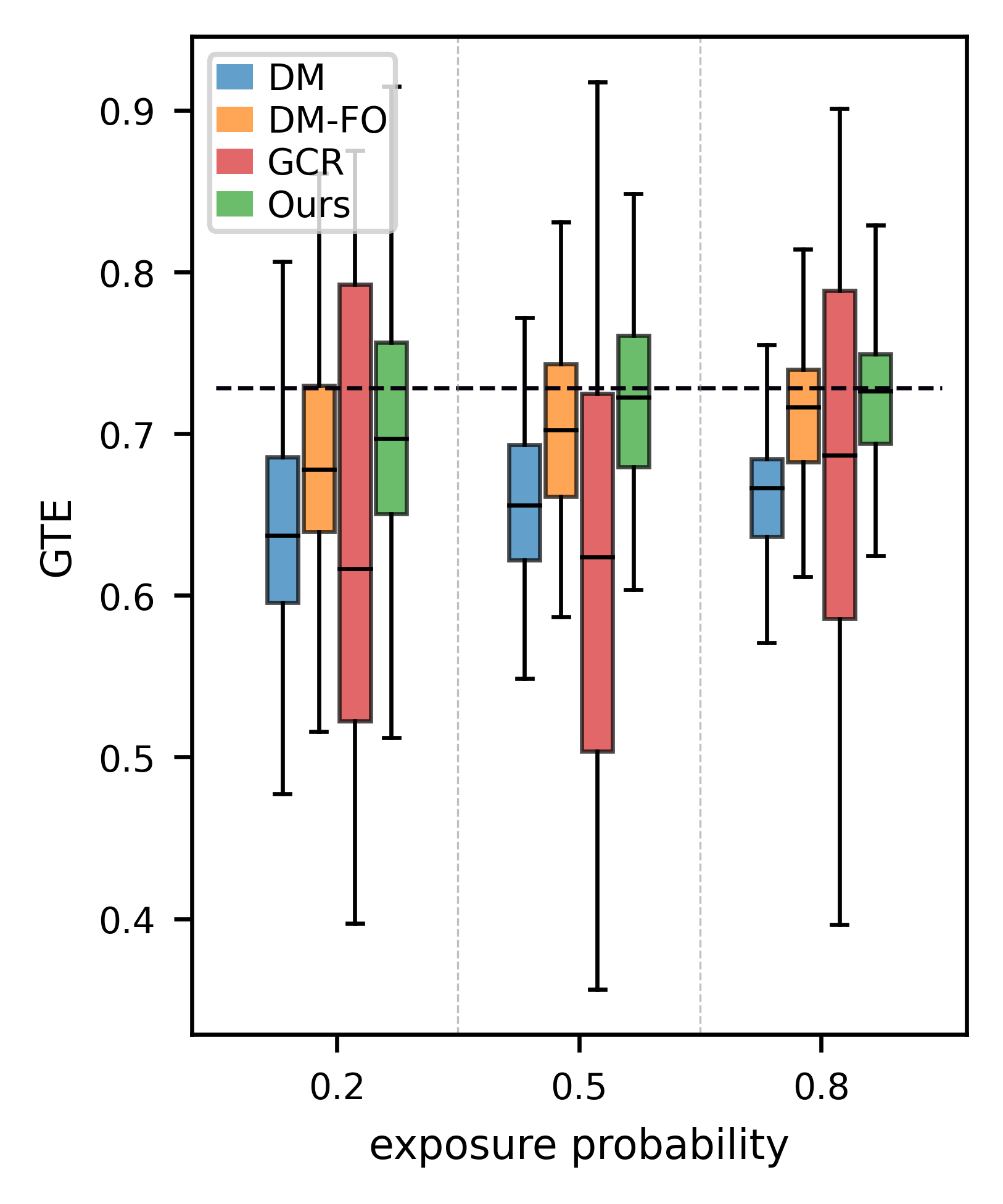}
    {\small (a) Boxplots of GTE estimates}
  \end{minipage}\hfill
  \begin{minipage}[t]{0.49\textwidth}
    \centering
    \includegraphics[width=\linewidth]{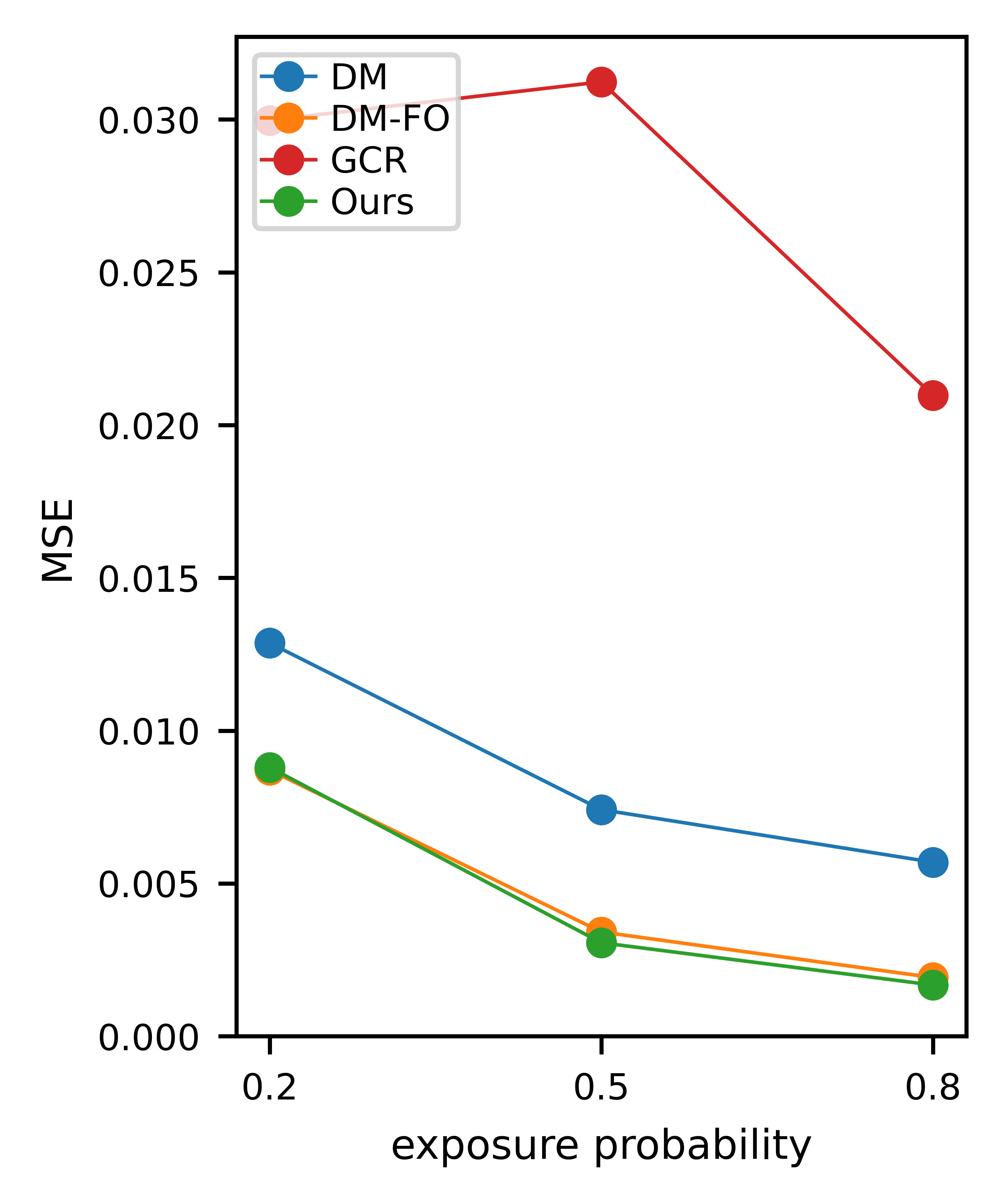}
    {\small (b) Mean squared error (MSE)}
  \end{minipage}
  \caption{Performance of different estimators for the impact of sharing under varying exposure probability.}
  \label{fig:barabasi-pi-GTE-estimate-mse}
\end{figure}

\subsection{Robustness to Stronger Receiver-Side Propagation}
\label{subsec:appendix-more-shares}

We further examine a setting with stronger receiver-side propagation. In this setting, the receiver-side propagation parameter $\varphi_j^{d,T}$ is sampled from $U(0,0.5)$ instead of $U(0,0.2)$, with $\varphi_j^{s,T}=\varphi_j^{d,T}$ as in the default setting. This creates a regime with stronger propagation effects and stronger interference. Table~\ref{tab:appendix-more-shares} reports the MSE of each estimator. The proposed estimator achieves the lowest MSE across all disturbance levels, and its advantage over all three baseline estimators becomes more pronounced in this stronger propagation setting.

\begin{table}[!htbp]
\centering
\small
\caption{MSE of different estimators under stronger receiver-side propagation}
\label{tab:appendix-more-shares}
\begin{tabular}{lccc}
\toprule
Method & $\Delta=0.3$ & $\Delta=0.5$ & $\Delta=0.7$ \\
\midrule
DM    & 0.2317 & 0.4213 & 0.6090 \\
DM-FO & 0.1044 & 0.1270 & 0.1770 \\
GCR   & 0.6451 & 0.7715 & 0.9011 \\
Ours  & \textbf{0.0880} & \textbf{0.0691} & \textbf{0.0650} \\
\bottomrule
\end{tabular}
\end{table}

\subsection{Robustness to Network Density and Scale}
\label{subsec:appendix-network-scale-degree}

We next vary the average degree and the number of users on the Barab\'asi--Albert graph to examine sensitivity to network density and scale. Tables~\ref{tab:appendix-ba-degree}--\ref{tab:appendix-ba-n} report the MSE of each estimator. The proposed estimator achieves the lowest MSE across all degree settings and remains best or nearly best across different network sizes. DM-FO is the closest competitor, while GCR is less stable.

\begin{table}[!htbp]
\centering
\small
\caption{MSE of different estimators with varying average degree.}
\label{tab:appendix-ba-degree}
\begin{tabular}{lccc}
\toprule
Method & $\bar g$=20 & $\bar g$=50 & $\bar g$=100 \\
\midrule
DM    & 0.0099 & 0.0074 & 0.0064 \\
DM-FO & 0.0042 & 0.0034 & 0.0026 \\
GCR   & 0.1921 & 0.0312 & 0.0055 \\
Ours  & \textbf{0.0036} & \textbf{0.0031} & \textbf{0.0022} \\
\bottomrule
\end{tabular}
\end{table}

\begin{table}[!htbp]
\centering
\small
\caption{MSE of different estimators with varying number of users.}
\label{tab:appendix-ba-n}
\begin{tabular}{lccc}
\toprule
Method & $n=10{,}000$ & $n=20{,}000$ & $n=50{,}000$ \\
\midrule
DM    & 0.0139 & 0.0085 & 0.0074 \\
DM-FO & 0.0113 & \textbf{0.0058} & 0.0034 \\
GCR   & 0.0154 & 0.0169 & 0.0312 \\
Ours  & \textbf{0.0112} & 0.0060 & \textbf{0.0031} \\
\bottomrule
\end{tabular}
\end{table}

\subsection{Robustness to Parameter Distributions}
\label{subsec:appendix-parameter-distributions}

The simulations in Section \ref{subsec:sim-setup} draw the base parameters $a_i^d$, $b_k^d$, $\phi_i^{d,T}$, $\varphi_j^{d,T}$, $\theta_{k}^{d,T}$ from uniform distributions. To check whether the results are driven by this specific distributional choice, we replace the uniform distributions with lognormal and gamma distributions calibrated to have the same means and variances. After drawing these base parameters, the remaining parameters are generated using the same transformations as in the default simulation setup. These alternatives allow a small fraction of users or contents to have much larger parameter values, thereby introducing stronger heterogeneity than the uniform specification. Tables~\ref{tab:appendix-lognormal}--\ref{tab:appendix-gamma} show that the qualitative ordering remains stable. This suggests that the performance gain is not driven by the specific uniform parameter specification.

\begin{table}[!htbp]
\centering
\small
\caption{MSE of different estimators under lognormal parameter distributions.}
\label{tab:appendix-lognormal}
\begin{tabular}{lccc}
\toprule
Method & $\Delta=0.3$ & $\Delta=0.5$ & $\Delta=0.7$ \\
\midrule
DM    & 0.0071 & 0.0083 & 0.0110 \\
DM-FO & 0.0059 & 0.0050 & 0.0057 \\
GCR   & 0.0139 & 0.0131 & 0.0157 \\
Ours  & \textbf{0.0057} & \textbf{0.0048} & \textbf{0.0048} \\
\bottomrule
\end{tabular}
\end{table}

\begin{table}[!htbp]
\centering
\small
\caption{MSE of different estimators under gamma parameter distributions.}
\label{tab:appendix-gamma}
\begin{tabular}{lccc}
\toprule
Method & $\Delta=0.3$ & $\Delta=0.5$ & $\Delta=0.7$ \\
\midrule
DM    & 0.0059 & 0.0084 & 0.0104 \\
DM-FO & \textbf{0.0045} & 0.0044 & 0.0046 \\
GCR   & 0.0379 & 0.0358 & 0.0367 \\
Ours  & \textbf{0.0045} & \textbf{0.0043} & \textbf{0.0039} \\
\bottomrule
\end{tabular}
\end{table}

\subsection{Robustness to Graph Topology}
\label{subsec:appendix-graph-topology}

We also consider graph topologies beyond the Barab\'asi--Albert graph. Specifically, we consider Watts--Strogatz graphs~\cite{watts1998collective} with different rewiring probabilities $p_{\mathrm{rew}}$ and stochastic block models~\cite{holland1983stochastic} with different numbers of blocks $K$. For the stochastic block graphs,
the between-block connection probability is set to be a fixed fraction of the
within-block connection probability, while the expected degree is kept fixed at the default value $\bar g=50$.
Thus, these experiments mainly vary the local clustering and community
structure of the network, rather than the overall graph density.
Tables~\ref{tab:appendix-ws}--\ref{tab:appendix-sbm} report the MSE of each estimator. The proposed estimator achieves the lowest or nearly lowest MSE in most settings. DM-FO is again a close competitor, while GCR performs well in some cases but is less stable across graph structures.

\begin{table}[!htbp]
\centering
\small
\caption{MSE of different estimators on Watts--Strogatz graphs with varying rewiring probability $p_{\mathrm{rew}}$.}
\label{tab:appendix-ws}
\begin{tabular}{lccc}
\toprule
Method & $p_{\mathrm{rew}}=0.05$ & $p_{\mathrm{rew}}=0.10$ & $p_{\mathrm{rew}}=0.30$ \\
\midrule
DM    & 0.00119 & 0.00113 & 0.00122 \\
DM-FO & 0.00064 & 0.00064 & 0.00072 \\
GCR   & 0.00084 & 0.00118 & \textbf{0.00065} \\
Ours  & \textbf{0.00058} & \textbf{0.00058} & 0.00068 \\
\bottomrule
\end{tabular}
\end{table}

\begin{table}[!htbp]
\centering
\small
\caption{MSE of different estimators on stochastic block models with varying numbers of equal-sized blocks $K$.}
\label{tab:appendix-sbm}
\begin{tabular}{lcccc}
\toprule
Method & $K=1$ & $K=3$ & $K=5$ & $K=10$ \\
\midrule
DM    & 0.00144 & 0.00118 & 0.00132 & 0.00103 \\
DM-FO & 0.00075 & 0.00078 & 0.00070 & 0.00051 \\
GCR   & 0.00248 & \textbf{0.00056} & 0.00097 & 0.00192 \\
Ours  & \textbf{0.00068} & 0.00073 & \textbf{0.00062} & \textbf{0.00048} \\
\bottomrule
\end{tabular}
\end{table}

\subsection{Sensitivity to Propagation Strength}
\label{subsec:appendix-rho0}

We further vary the overall strength of network propagation. We summarize
propagation strength by a target spectral radius parameter $\rho_0$, computed
from the sharing probability matrices used in the simulation. In the
implementation, we fix the graph and the sampled model parameters, and vary
$\rho_0$ by uniformly rescaling the sharing parameters before running the
simulation. This changes the overall propagation strength while preserving the
relative heterogeneity across users, contents, and treatment groups. Implementation details are provided in the released code. A larger $\rho_0$
corresponds to stronger propagation and moves the system closer to the
near-critical regime. Because this sweep is computationally more expensive, we
use a smaller setting with $n=5{,}000$, $m=200$, and fix $\Delta=0.5$.

Table~\ref{tab:appendix-rho0} shows that the MSE of all estimators increases sharply as $\rho_0$ becomes larger, reflecting the increasing difficulty of estimation under stronger network amplification. DM-FO performs best when $\rho_0$ is low or moderate, while the proposed estimator remains very close. When propagation is strongest, at $\rho_0=0.8$, the proposed estimator achieves the lowest MSE. These results suggest that our estimator remains competitive under moderate propagation and can be especially useful when network amplification is strong.

\begin{table}[!htbp]
\centering
\small
\caption{Varying target spectral radius $\rho_0$ at $n=5{,}000$, $m=200$}
\label{tab:appendix-rho0}
\begin{tabular}{lccc}
\toprule
Method & $\rho_0=0.2$ & $\rho_0=0.5$ & $\rho_0=0.8$ \\
\midrule
DM    & 0.0580 & 2.0642 & 49.5012 \\
DM-FO & \textbf{0.0507} & \textbf{1.2891} & 30.5915 \\
GCR   & 0.0835 & 2.0271 & 45.6366 \\
Ours  & 0.0519 & 1.3090 & \textbf{27.9928} \\
\bottomrule
\end{tabular}
\end{table}

\end{document}